\documentclass[times,sort&compress,3p]{elsarticle}
\journal{Journal of Multivariate Analysis}
\usepackage[labelfont=bf]{caption}

\usepackage{amsmath,amsfonts,amssymb,amsthm,booktabs,color,epsfig,graphicx,hyperref,url}

\usepackage{adjustbox}

\theoremstyle{plain}
\newtheorem{theorem}{Theorem}

\newtheorem{lemma}{Lemma}
\newtheorem{corollary}{Corollary}

\theoremstyle{definition}
\newtheorem{definition}{Definition}

\usepackage{xcolor}

\newcommand\revision[1]{{\color{black}{{#1}}}}

\newcommand\minorrevision[1]{{#1}}

\newcommand\varnum{r}

\usepackage[ruled, linesnumbered]{algorithm2e}
\SetKwInOut{Parameter}{Parameters}
\SetKwRepeat{RepeatUntil}{repeat}{until}

\DeclareMathOperator{\tr}{tr}

\makeatletter
\def\ps@pprintTitle{%
  \let\@oddhead\@empty
  \let\@evenhead\@empty
  \let\@oddfoot\@empty
  \let\@evenfoot\@oddfoot
}
\makeatother

\begin{document}

\begin{frontmatter}

\title{A method for sparse and robust independent component analysis}

\author[1]{Lauri Heinonen\corref{mycorrespondingauthor}}
\author[1]{Joni Virta}

\address[1]{Department of Mathematics and Statistics, University of Turku, 20014 Turku, Finland}

\cortext[mycorrespondingauthor]{Corresponding author. Email address: \url{lauri.k.heinonen@utu.fi}}

\begin{abstract}

This work presents sparse invariant coordinate selection, SICS, a new method for sparse and robust independent component analysis. SICS is based on classical invariant coordinate selection, which is presented in such a form that a LASSO-type penalty can be applied to promote sparsity. Robustness is achieved by using robust scatter matrices. In the first part of the paper, the background and building blocks: scatter matrices, measures of robustness, ICS and independent component analysis, are carefully introduced. Then the proposed new method and its algorithm are derived and presented. This part also includes consistency and breakdown point results for a general case of sparse ICS-like methods. The performance of SICS in identifying sparse independent component loadings is investigated with \revision{multiple} simulations. The method is illustrated with an example in constructing sparse causal graphs \revision{and we also propose a graphical tool for selecting the appropriate sparsity level in SICS}.
\end{abstract}





\end{frontmatter}

\section{Introduction}\label{sec:introduction}








A common goal in statistics and data science is to break down the variation in data to some factors. In this article this problem is tackled with sparse independent component analysis where the variables in data are decomposed into a linear combination of independent latent components (sources). The components are built so that only a subset of the original variables are linked to a particular component (sparsity). Reasons for sparsity are easier interpretability, as there are less coefficients to interpret, and trying to avoid overfitting.

In independent component analysis (ICA) the observed random vector $x$ is assumed to be an unknown deterministic mix of unknown random and independent sources \citep{hyvarinen2001independent}. This (basic linear) independent component model can be written as
\begin{align*}
     x = \Omega z + \mu,
\end{align*}
where $x$ is the observed vector, $\mu \in \mathbb{R}^p$ is some deterministic mean, $\Omega \in \mathbb{R}^{p \times p}$ is the mixing matrix and $z$ is the source vector with independent components. The objective of ICA is to, given a random sample $x_1, \ldots, x_n$ from the model, find an unmixing matrix (estimate of $\Omega^{-1}$) that can be used to obtain the values of the independent sources from the observations. Many different ICA methods have been proposed, for example JADE \citep{cardoso1993JADE} and FastICA \citep{hyvarinen1999fastICA}, see also the review in \cite{nordhausen2018independent}.

Our chosen way to solve the ICA problem comes from the invariant coordinate selection (ICS) \citep{tyler2009invariant}. In ICS, one has two scatter matrices $S_1$ and $S_2$, which are in certain sense generalizations of the covariance matrix. An example of such a matrix is the FOBI-matrix based on the fourth moments, see Section \ref{sec:ica_and_scatter} for details. Now, one jointly diagonalizes the matrices by solving $S_2 v = \lambda S_1 v$. This procedure essentially first removes the variation measured by $S_1$ from the data and then finds a coordinate system (of which $v$ is a basis vector) that maximizes the variation measured by $S_2$. This way we can, in a sense, ``compare'' the variations measured by the scatter matrices, calculated on the same data, and get a new view on the data which best contrasts these two forms of variation, see \cite{tyler2009invariant} for examples. Several sub-cases of ICS have separate names and one example of such a method is the classical ICA method FOBI \citep{cardoso1989FOBI}, where the scatter matrices are the covariance matrix and the FOBI-matrix. Also ICS with other scatter matrices can lead to the solution of the ICA problem and we discuss this relation closer in Section \ref{sec:ica_and_scatter}.

The purpose of this work is to develop, using ICS as a basis, a general methodology for sparse and robust independent component analysis, where by ``robust'' we mean that the method does not break down under the presence of outlying data points. Both topics have been separately pursued in the literature and we next review the most relevant works but, as far as we are aware, the combination of sparsity and robustness is entirely novel. 

Independent component analysis and sparsity can be combined in two natural, but fundamentally different, ways. In \textit{sparse component analysis}, the source vectors $z$ are assumed to be sparse, see, e.g., \cite{malioutov2005sparse, babaie2006sparse, georgiev2007sparse, boukouvalas2018sparsity}. Sparse component analysis is at its most useful in specialized applications where we have reason to expect that a certain subset of the sources is ``inactive'' (taking only the value zero) at any given time. Whereas, in \textit{sparse independent component analysis}, the unmixing matrix $\Omega^{-1}$ (or, sometimes, its inverse) is assumed to be sparse. In this work, we focus on this variant, sparse independent component analysis, due to its generality; the sparsity of $\Omega^{-1}$ leads to highly interpretable independent components, making the method a useful tool regardless of the application area.

Several non-robust sparse ICA methods have been proposed earlier in the literature, the most prevalent approach being formulating a likelihood function for the observations and using a sparsity-inducing penalty function: \cite{hyvarinen2002imposing} took a Bayesian approach and formulated the penalization through conjugate priors, allowing fitting their model via standard ICA applied to an augmented sample; \cite{zhang2006ica} used a SCAD-penalty in combination with post-estimation thresholding; \cite{zhang2009ica} proposed using either adaptive LASSO or optimal brain surgeon, choosing their tuning parameters using AIC/BIC; \cite{palsson2014sparse} used directly a $\ell_0$-penalty; \cite{chen2019sparse} used a combination of standard and group LASSO penalties; \cite{harada2020estimation} used adaptive LASSO with a relaxed form of orthogonality; \cite{ng2023identifiability} maximized a constrained Gaussian likelihood under specific structural assumptions on $\Omega$. Outside of likelihood-based approaches, \cite{abrahamsen2018sparse} derived high-dimensional error bounds under a specific Gaussian form of ICA where the sparsity of $\Omega$ ensures the identifiability of the model, and \cite{ng2023identifiability} proposed a second-order decomposition method for a specific class of structured mixing matrices $\Omega$. \revision{Table \ref{tab:sICA_comparison} in \ref{sec:comparison_table} presents a more detailed comparison of these methods' properties.}

In this work, we take a different viewpoint from the above and implement sparsity using the framework in \cite{li2007sparse} that was inspired by the seminal LASSO-based sparse PCA method by \cite{zou2006sparse}. While \cite{li2007sparse} focused exclusively on sparse sufficient dimension reduction, their framework is directly applicable to the scatter matrix formulation of ICA (and also ICS in general), a fact that appears not have been noticed earlier in the literature. Using this approach to achieve sparse and robust ICA is natural, as the robustness can be implemented to the procedure via the choice of scatter matrices, essentially separating the two aspects (sparsity and robustness) and allowing controlling them through individual tuning parameters.

Of earlier approaches to robust ICA, our work is most similar with \cite{nordhausen2008robust} who likewise used robust scatter matrices to achieve robust estimation of independent components. Besides this, the previous literature on robust ICA includes maximizing robust measures of shape \citep{baloch2005robust} or divergence \citep{chen2013robust}, and using rank \citep{ilmonen2011semiparametrically} and signed rank \citep{hallin2015r} based estimators to circumvent moment assumptions and to achieve semiparametric efficiency. None of the previous references allow for sparsity in the estimation of the independent components.

The main contributions and novelty in the current paper are the following:
\begin{itemize}
    \item We propose a general framework of robust and sparse ICA. Unlike many of the previous works on sparse ICA \citep{hyvarinen2002imposing, zhang2006ica, zhang2009ica, palsson2014sparse, chen2019sparse, harada2020estimation}, our model is semiparametric (and not likelihood-based), meaning that we require neither the pre-specification nor the estimation of the densities of the latent sources. Moreover, both the level of sparsity and the level of robustness of the estimation are controlled via individual tuning parameters, making the method transparent and simple to use in practice.
    \item We establish the asymptotic convergence rate of the resulting estimator as a function of the convergence rates of the used scatter matrices. This result is not limited to ICA but actually applies to the joint diagonalization of any two scatter matrices. As such, it gives convergence rates also in the sufficient dimension reduction context \cite{li2007sparse} and in the fully general context of sparse ICS (SICS). Moreover, when compared to earlier works on asymptotics of sparse ICA \cite{abrahamsen2018sparse, chen2019sparse}, our results do not make distributional assumptions and allow robustness.
    \item \minorrevision{We characterize the breakdown point of SICS in terms of the breakdown points of the two scatter matrices.}
    \item We extensively study the impact of sparsity and robustness on the finite-sample properties of the estimator using simulations. The results confirm that robust methods are needed when the data is contaminated and that sparse methods are clearly beneficial when the underlying situation is sparse. The results also suggest that the sparse and robust variants of the proposed method generally perform extremely well in the tested scenarios.
    \item \revision{We propose a visual tool that allows selecting the appropriate sparsity level in a data-driven way and apply it to a dataset on diabetes patients. The tool is inspired by the classical stability selection \cite{meinshausen2010stability}.}
    
\end{itemize}


In Section~\ref{sec:ica_and_scatter}, we review the concepts of a scatter matrix and independence property, and recall how they produce a solution to the ICA problem. Also several robust scatter matrices from the literature are presented. Section~\ref{sec:main_method} sees us combining robust scatter matrices with the sparsity framework of~\cite{li2007sparse} to obtain our proposed method. We also establish the convergence rate and the breakdown point of the method in this section. In Section~\ref{sec:simulations}, we compare the performance of the method under different levels of sparsity and robustness using simulations, \revision{including also several competing methods in the evaluations.} A real data example of estimating non-Gaussian causal graphs using sparse ICA is given in Section~\ref{sec:real_data}. \revision{A graphical tool for selecting the sparsity level is proposed in Section \ref{sec:selection_tool},} and we conclude with discussion in Section~\ref{sec:discussion}. \minorrevision{The following notation will be used for different norms: $\| \cdot \|$ denotes the Euclidean vector norm, $\| \cdot \|_1$ the $L_1$-norm for vectors, $\| \cdot \|_F$ the Frobenius norm for matrices and $\| \cdot \|_2$ the spectral norm for matrices.}



\section{Theoretical background}\label{sec:ica_and_scatter}


\subsection{ICA solution using two scatter matrices}

First let us define a location vector, a scatter matrix, the independence property and symmetrized scatter matrices.
\begin{definition}
Let $x \in \mathbb{R}^{p}$ be a random vector. A vector $T \in \mathbb{R}^p$, calculated from $x$, is a \emph{location vector} if it is affine equivariant, in the sense that
\[T(Ax+b)=AT(x)+b\]
for all full rank matrices $A \in \mathbb{R}^{p{\times}p}$ and vectors $b \in \mathbb{R}^p$.
\end{definition}
\begin{definition}
Let $x \in \mathbb{R}^{p}$ be a random vector. A matrix $S \in \mathbb{R}^{p{\times}p}$, calculated from $x$, is a \emph{scatter matrix} if it is positive definite and affine equivariant, in the sense that
\[S(Ax+b)=AS(x)A^\top\]
for all full rank matrices $A \in \mathbb{R}^{p{\times}p}$ and vectors $b \in \mathbb{R}^p$.
\end{definition}
Scatter matrices are a generalization of the regular covariance matrix $\Sigma := \mathrm{Cov}(x)$, which clearly satisfies the desired equation. Some scatter matrices, for example the covariance matrix, have an additional independence property.
\begin{definition}
A scatter matrix $S$ is said to have the \emph{independence property} if $S(x)$ is diagonal for all random vectors $x$ with independent components.
\end{definition}
Now, a question rises, how can we find scatter matrices with the independence property? Besides the covariance matrix, another well-known example is the FOBI-matrix (used in an ICA method named fourth order blind identification \cite{cardoso1989FOBI}) $S_{\mathrm{FOBI}} = \mathbb{E}[ \Tilde{x} \Tilde{x}^\top \Sigma^{-1} \Tilde{x} \Tilde{x}^\top]$, where $\Tilde{x} := x - \mathbb{E}(x)$ and $\Sigma$ denotes the covariance matrix of $x$. Moreover, further scatter matrices with the independence property can be obtained via the process call symmetrization, and we next state \citep[Theorem 1]{oja2006scatter} defining \emph{symmetrized scatter matrices}
\begin{theorem}
\citep{oja2006scatter} Let $S$ be a scatter matrix and $x_1$, $x_2$ two independent copies of a random vector $x$. Then the symmetrized scatter matrix $S_s(x) := S(x_1-x_2)$ has the independence property.
\end{theorem}
On the sample level, the symmetrized scatter matrix $S_s$ can be calculated by applying $S$ to the sample of all $n^2$ pairwise differences $x_i-x_j$ (where $i,j=1,2,\dots,n$) of the original observations. \revision{As remarked in \cite{nordhausen2015cautionary}, apart from symmetrized scatter matrices, no other scatter matrices are known to possess the independence property, see also \cite{virta2016one}. Note that $\mathrm{Cov}(x)$ and $S_{\mathrm{FOBI}}$ can be written in a symmetrized form, see \cite{nordhausen2015cautionary}.}

Next, we define the independent component (IC) model using a general location vector and scatter matrix.
\begin{definition}\label{def:ic_model}
Let $T$ be a location vector and let $S_1$ and $S_2$ be two scatter matrices with the independence property. A random vector $x \in \mathbb{R}^p$ is said to have IC-model with respect to $T$, $S_1$ and $S_2$ if
\[
x = \Omega z + \mu,
\]
where the random vector $z \in \mathbb{R}^p$ has independent components, $T(z)=0$, $S_1(z)=I_p$ and $S_2(z)$ is a diagonal matrix with elements $d_1 > \dots > d_p > 0$.
\end{definition}
One could relax the assumptions to allow non-strict inequalities for the $d_j$, but this would lead to non-identifiable components, and we assume their strictness throughout this work.
Now \citep[Theorem 2]{oja2006scatter} gives us a way to perform ICA, i.e., to estimate the source vector $z$, in such a case using matrix decompositions. This procedure is equivalent to ICS \cite{tyler1987distribution} with the choice $S_1, S_2$ of scatter matrices.

\begin{theorem}\label{theo:ic_solution_1}
\citep{oja2006scatter} Let $x$ have an IC-model with respect to $T$, $S_1$ and $S_2$. Let us define
\[
B(x) = S_1(x)^{-1/2} \left[U_2\left(S_1(x)^{-1/2}x\right)\right],
\]
where $U_2(x)$ is the matrix of eigenvectors of $S_2(x)$ (in the order of decreasing eigenvalues). Then
\[
B(x)^\top\{ x-T(x) \}=Jz
\]
for some a diagonal matrix $J$ with elements $\pm 1$.
\end{theorem}

In Theorem \ref{theo:ic_solution_1}, the location vector $T$ is used simply to fix the location of the independent components. Even without $T$, the ICs could be estimated up to location as $B(x)^\top x$ and, as such, we focus solely on the scatter matrices $S_1$ and $S_2$ in the following, disregarding the location estimation.
Independent component analysis based on Theorem~\ref{theo:ic_solution_1} using symmetrized scatter matrices is discussed in \citep{taskinen2007independent}. Finally, we note that Theorem \ref{theo:ic_solution_1} indeed states that any two scatter matrices with the independence property can be used to solve the IC problem on the population level. However, these scatter matrices might involve various assumptions and their finite-sample properties can still be different. Later we focus on a particular class of scatters that allows us to solve the problem in an outlier-resistant way.





\subsection{Regression formulation for scatter matrix ICA}

Fix next two scatter matrices, $S_1$ and $S_2$, both having the independence property. Computing the respective IC solution is simple to do via eigendecompositions as specified in Theorem \ref{theo:ic_solution_1}, but we next present still an alternative way of obtaining the same solution, originally presented in the context of sufficient dimension reduction in \cite{li2007sparse}. The reason for introducing this auxiliary (and more complex) way of obtaining the solution has the benefit that it can be combined with sparsity in a natural way.

For $k \in \{ 1, \ldots, p \}$, let $U_k \in \mathbb{R}^{p \times k}$ denote the matrix comprising of the first $k$ columns of $U_2\left(S_1(x)^{-1/2}x\right)$.

\begin{theorem}\label{theo:ic_solution_2}
    Let $r_j \in \mathbb{R}^p$ denote the columns of $S_2(x)^{1/2}$. Then, the minimizers of
    \begin{align*}
        \sum_{j = 1}^p \| S_1(x)^{-1/2} r_j - A B^\top r_j\|^2
    \end{align*}
    over $A, B \in \mathbb{R}^{p \times k}$, $A^\top A = I_k$, are precisely the pairs $(A, B) = (U_k O_k, S_1(x)^{-1/2} U_k O_k)$, where $O_k$ is any $k \times k$ orthogonal matrix.
\end{theorem}

\begin{proof}[Proof of Theorem \ref{theo:ic_solution_2}]
    Simplifying the objective function, we see that the problem is equivalent to minimizing 
    \begin{align*}
        -2 \mathrm{tr}(B^\top S_2(x) S_1^{-1/2}(x) A) + \mathrm{tr}(B^\top S_2(x) B).
    \end{align*}
    Differentiating this with respect to $B$, we get the gradient $-2 S_2(x) S_1^{-1/2} A + 2 S_2 B$. The positive definiteness of $S_2$ thus shows that the minimizing value of $B$ satisfies $B = S_1^{-1/2}(x) A$. Plugging this back in to the objective function, we see that the optimal $A$ maximizes the map $A \mapsto \mathrm{tr}(A^\top S_1^{-1/2}(x) S_2(x) S_1^{-1/2}(x) A)$
    over $A \in \mathbb{R}^{p \times k}$, $A^\top A = I_k$. By \cite[Corollary 4.3.39]{horn2012matrix}, all maximizers are thus of the form $A = V_k O_k$ where $V_k \in \mathbb{R}^{p \times k}$ contain first $k$ eigenvectors of $S_1^{-1/2}(x) S_2(x) S_1^{-1/2}(x)$ as its columns and $O_k$ is an arbitrary $k \times k$ orthogonal matrix. By the affine equivariance of scatter matrices, $U_2\left(S_1(x)^{-1/2}x\right)$ in Theorem \ref{theo:ic_solution_1} contains the $p$ eigenvectors of $S_2(S_1^{-1/2}(x) x) = S_1^{-1/2}(x) S_2(x) S_1^{-1/2}(x)$, showing that $V_k = U_k$ (up to sign), and concluding the proof.
\end{proof}

Two notes are in order: (i) Theorem \ref{theo:ic_solution_2} essentially converts the IC problem into minimization of a specific sum of squares. In the next section, we show how this minimization can be carried out using alternative least squares, allowing the sparsification of the solution through an $\ell_1$-penalty. (ii) All minimizers $B$ of the objective function in Theorem \ref{theo:ic_solution_2} satisfy $\mathrm{col}(B) = \mathrm{col}(B(x))$ where $\mathrm{col}(\cdot)$ denotes the column space and $B(x)$ is as in Theorem \ref{theo:ic_solution_1}. As such, the solution in Theorem~\ref{theo:ic_solution_2} is slightly less general than the original IC solution in Theorem~\ref{theo:ic_solution_1} in that the former captures only the space spanned by the first $k$ independent components, not their individual directions. Interestingly, this ambiguity is unavoidable and occurs even though our ICs have differing ``kurtoses'' $d_1, \ldots, d_p$ in Definition \ref{def:ic_model}. Naturally, this ambiguity vanishes when $k = 1$ (the only $1 \times 1$ orthogonal matrices are the scalars $\pm 1$), and to avoid it also in the case $k > 1$, later in Section \ref{sec:main_method}, we propose a correction (lines 12--13 in Algorithm \ref{alg:SICS}) that corrects for the presence of the orthogonal transformation $O_k$ in our algorithm.

Finally, we remark that this $O_k$-ambiguity applies also to Theorem 3 in \cite{zou2006sparse} and Propositions 1 and 2 in \cite{li2007sparse} but, likely due to an oversight, the authors have missed it. However, this omission was noted later in Remark II.1 in \cite{deng2019group}.

\subsection{Symmetrized robust scatter matrices}

Having observed how the ICA solution can be found using two scatter matrices, we now take a look at a specific class of robust scatter matrices, robust $M$-estimators. By robust we mean that small deviations from assumptions do not impair the model's performance too much and large deviations from model do not cause a catastrophe \citep{huber2011robust}. For us, the interesting deviation from assumptions is the existence of outliers. Since we will eventually focus exclusively on symmetrized scatter matrices, we assume, without loss of generality, that the true location of the data $x_1, \ldots, x_n$ is zero.

Now, given a function $\rho: \mathbb{R} \rightarrow \mathbb{R}$, the corresponding $M$-estimator scatter matrix $S$ is found by minimizing
\[
L(S) = \frac{1}{n} \sum_{i=1}^n \{ \rho(x_i^\top S^{-1}x_i) - \rho(x_i^\top x_i) \} + \ln\det S.
\]
This expression originally comes from the maximum likelihood estimation of an elliptic distribution with density proportional to $\mathrm{det}(S^{-1})\exp(-\rho(x^\top S^{-1}x))$. The function $\rho$ is typically assumed to satisfy various regularity conditions, see \cite{dumbgen2015m, dumbgen2016new}. One common choice for the function $\rho$ is $\rho_{\nu,p}(z) = (\nu+p)\ln(\nu+z)$ which corresponds to the density of the $p$-variate $t$-distribution with $\nu>0$ degrees of freedom. 



The above minimization problem leads to an estimating equation
\begin{align}\label{eq:estimating_eq}
\frac{1}{n} \sum_{i=1}^n w\left(\|S^{-1/2}x_i\|^2\right) \|S^{-1/2}x_i\|^{-2} S^{-1/2}x_ix_i^\top S^{-1/2}
= I_p,    
\end{align}
where $w(z)=\rho'(z)z$ is called a weight function. $M$-estimators are often defined in terms of the function $w$ instead of $\rho$. 
Common choices for $w$ are \citep{sirkia2007symmetrised} $w(z)=z$ giving the regular covariance matrix, $w(z)=p$ giving Tyler's $M$-estimator \citep{tyler1987distribution}, and Huber's $M$-estimator which uses $w(z) = \mathbb{I}(|z| \leq c) |z|/\sigma + \mathbb{I}(|z| > c) c/\sigma$,
where the tuning constant $c$ is chosen to control $\mathbb{P}(\chi^2_p \leq c^2/2)$ and the scaling factor $\sigma$ is chosen so that $\mathbb{E}[w(\|x\|^2)]=p$ for $x \sim N(0,I_p)$.

Now we can define symmetrized $M$-estimates by replacing the observations $x_i$ by pairwise differences $x_i-x_j$ (similarly, the averaging $\frac{1}{n}\sum_{i=1}^n$ in \eqref{eq:estimating_eq} is replaced by the double sum $\frac{2}{n(n-1)}\sum_{i<j}$). In this work, we use symmetrized versions of two of the robust scatter matrix families discussed above, the MLE of $t$-distribution \citep{dumbgen2016new} and Huber's M-estimator \citep{sirkia2007symmetrised} \revision{as they fulfill the two important preconditions: the independence property for doing ICA with ICS, and robustness}. These have been implemented in the R-packages \texttt{fastM} \citep{RfastM} and \texttt{ICSNP} \citep{RICSNP}, respectively.

\revision{Naturally, other choices of symmetrized scatters could also be equally viable in our context. Scatter matrices can be categorized into three groups \citep{tyler2009invariant}: Class I contains all non-robust scatters (such as $\mathrm{Cov}(x)$ and $S_{\mathrm{FOBI}}$), Class II contains scatters that are moderately robust and have relatively light computational complexity (such as $M$-estimators), and Class III contains scatters with high robustness and high computational complexity (such as the minimum covariance determinant matrix \citep{rousseeuw1985multivariate}). \cite[Section 6]{tyler2009invariant} state that using two Class II matrices is typically a reasonable choice in ICS when one wants to avoid the computational load of Class III scatters, which in our case would still be increased due to symmetrization. And, while class II contains also other matrices besides our pair of choice, such as D{\"u}mbgen's estimator \citep{dumbgen1998tyler} or one-step $W$-estimators \citep{tyler2009invariant}, our primary focus is on sparsity and thus we have not attempted to perform any extensive evaluations between different robust scatters here. However, several comparisons of symmetrized robust scatter matrices can be found in the literature, see \cite{sirkia2007symmetrised, taskinen2007independent, miettinen2016computation}.}



\subsection{Measuring robustness}

The two most common ways to measure robustness of a statistic are the breakdown point and influence function. Breakdown point describes how large proportion of observations can be outliers without the statistic having arbitrary large/small values. Let $T$ be a statistic taking values in some normed space $(\mathcal{T}, \| \cdot \|_{\mathcal{T}})$, $X$ a sample and $X_\varepsilon$ a sample with proportion $\varepsilon$ of the observation replaced with arbitrary values, i.e., an $\varepsilon$-contamination of $X$. Let 
\[
b(X_\varepsilon,T) = \sup_{X_\varepsilon}\|T(X)-T(X_\varepsilon)\|_{\mathcal{T}}
\]
be the maximal bias. Now we can define the finite-sample breakdown point of $T$ as in \citep{donoho1983notion}, as
\[
\varepsilon^* = \inf \{ \varepsilon \mid b(X_\varepsilon,T)=\infty \}.
\]
As an example of the breakdown point in the case of simple estimators, consider the sample mean and sample median. Sample mean has the lowest possible breakdown point because replacing just one observation with arbitrarily high values can increase the mean without any limit. On the other hand, to increase the sample median to arbitrarily high value requires replacing half of the sample, making its breakdown point $1/2$.

It is known that, in general, $M$-estimators have no greater breakdown points than $1/(p+1)$ \citep{maronna1976robust}. Still \cite{tyler2014breakdown} shows that when the contaminating points are assumed to not lay in any low-dimensional hyperplane (called coplanar contamination), $M$-estimates can have a breakdown point close to $1/2$. In general symmetrization can be expected to lower the breakdown point of a scatter matrix, see \cite{dumbgen2005breakdown}, but, as our simulations later in Section \ref{sec:simulations} show, robust M-estimators, even when symmetrized, tolerate outliers particularly well.

Another measure of robustness is the influence function for a statistic $T$. Let $Q$ be a probability distribution and $Q_\varepsilon = (1-\varepsilon) Q + \varepsilon \Delta_x$ be an $\varepsilon$-contaminated version of $Q$, where $\Delta_x$ is a distribution with all probability mass concentrated to the point $x$. Then the influence function of $T$ for a point $x$ can be defined as
\[
IF(x;T,Q) = \lim_{\varepsilon \to 0} \frac{T(Q_\varepsilon)-T(Q)}{\varepsilon}.
\]
The influence function measures how much a (infinitely) small change in the distribution at a point $x$ can change the value of the statistic $T$. One way to characterize a robust statistic is to require that it has a bounded influence function which means that no single contamination at any point can influence the statistic by an unlimited amount.

It follows from Theorem 3 in \cite{sirkia2007symmetrised} that symmetrized $M$-estimator has a bounded influence function when the weight function $w$ is bounded. This is the case for symmetrized Huber's $M$-estimator and the MLE of $t$-distribution with $w(z)=\rho'(z)z=(\nu+p)(1+\frac{\nu}{z})$.

\section{Sparse and robust ICA}\label{sec:main_method}

\subsection{Algorithm}

Based on the concepts in the previous section, we next propose an algorithm for obtaining a robust and sparse ICA solution. Let $S_1$ and $S_2$ be two robust scatter matrices with the independence property, $r_j \in \mathbb{R}^p$ denote the columns of $S_2(x)^{1/2}$ and $\lambda_m \geq 0$ be the sparsity penalization parameter corresponding to the $m$th independent component. Then, the matrix $B = (\beta_1,\dots,\beta_k)$ part of the minimizer  $A,B \in \mathbb{R}^{p \times k}$ $(A^\top A~=~I_k)$ of
\begin{align}\label{eq:main_optimization}
    \sum_{j = 1}^p \| S_1(x)^{-1/2} r_j - A B^\top r_j\|^2 + \sum_{m=1}^{k} \lambda_m \|\beta_m\|_1
\end{align}
is the sparse invariant coordinate selection estimate, up to multiplying from right by some $k \times k$ orthogonal matrix $O_k$ (see Theorem \ref{theo:ic_solution_2}).  In practice, one usually directly chooses the desired numbers of non-zero components $\varnum$ (usually taken to be the same for all components) and the penalization parameters $\lambda_j$ are then determined implicitly based on~$\varnum$.

\begin{algorithm}
\caption{Sparse invariant coordinate selection}\label{alg:SICS}
\KwData{Scatter matrices $S_1, S_2 \in \mathbb{R}^{p{\times}p}$}
\Parameter{Number of invariant coordinates $k$\\
The numbers of non-zero coordinates $(r_1,\dots,r_k)$}
\KwResult{Matrix of the coefficient vectors $B=(\beta_1,\dots,\beta_k) \in \mathbb{R}^{p{\times}k}$}
Let $A =(\alpha_1,\dots,\alpha_k) \in \mathbb{R}^{p{\times}k}$ be the usual non-penalized ICS estimate with respect to $S_1$ and $S_2$; \\
$B = S_1^{-\frac{1}{2}}A$; \\
$B=\text{fix\_signs}(B)$; \\
\RepeatUntil{$\|B-B_{temp}\|_{F}<10^{-6}$}{
$B_{temp}=B$; \\
\For{$j \in \{1,2,\dots,k\}$}{
$\beta_j = \arg\min_{\beta_j} \left\{\sum \|S_2^{\frac{1}{2}}S_1^{-\frac{1}{2}}\alpha_j - S_2^{\frac{1}{2}}\beta_j\|^2 + \lambda_j\| \beta_j\|_1\right\}$;
}
$B=\text{fix\_signs}(B)$; \\
Calculate the singular value decomposition $S_1^{-\frac{1}{2}} S_2 B = UDV^\top$; \\
$A_0=UV^\top$; \\
Calculate the eigendecomposition $A_0^\top S_1^{-\frac{1}{2}} S_2 S_1^{-\frac{1}{2}} A_0 = O_k\Delta O_k^\top$; \\
$A=A_0 O_k$;
}
\Return{$B$};
\end{algorithm}

A procedure for minimizing \eqref{eq:main_optimization} is shown in Algorithm \ref{alg:SICS}. The command fix\_signs (lines 3, 9) flips the signs of the columns $\beta_j$ of $B$ so that the first non-zero element in each column is positive. This is done to eliminate the effect of signs changing between iterations of the repeat-until-loop. Line number 7 comes from being able to represent \eqref{eq:main_optimization} (for a fixed $A$) as
\[
\tr((I_p-AA^\top)S_1^{-1}S_2) + \| S_2^{\frac{1}{2}}S_1^{-\frac{1}{2}}A - S_2^{\frac{1}{2}}B \|_F^2 + \sum_{m=1}^{k} \lambda_m \|\beta_m\|_1,
\]
and we use the function \texttt{solvebeta} from the R-package \texttt{elasticnet} \citep{Relasticnet} to solve this LASSO-type problem (this function allows us to choose the number of non-zero components $\varnum$ instead of the penalization parameters $\lambda_j$).

The estimate for $A$ calculated on lines 10--11 comes from representing \eqref{eq:main_optimization} (for fixed $B$, so without needing the penalty term) as
\begin{align*}
    \sum_{j = 1}^p \| S_1^{-1/2} r_j - A B^\top r_j\|^2 = \| S_1^{-1/2} S_2^{1/2} - A B^\top S_2^{1/2}\|_F^2 = \| S_2^{1/2} S_1^{-1/2} - S_2^{1/2} S_1^{-1/2} S_1^{1/2} B A^\top\|_F^2
\end{align*}
and using the reduced rank Procrustes rotation \citep[Theorem 4]{zou2006sparse} to minimize this expression with respect to $A$. As this estimate is only calculated up to the rotation $O_k$, we choose $O_k$ (lines 12--13) so that $A$ diagonalizes $S_1^{-\frac{1}{2}} S_2 S_1^{-\frac{1}{2}}$. This choice thus corresponds to how the solution to the joint diagonalization behaves in the regular ICS. 

Algorithm \ref{alg:SICS} thus produces a sparse matrix $B$ (whose sparsity level is controlled by the parameters $\varnum_1, \ldots, \varnum_k$), using which the independent components are obtained as $B^\top x_i$.

\subsection{Consistency of SICS}

We next establish conditions under which the proposed sparse and robust ICA yields consistent estimates of the independent components. The result is stated in terms of two arbitrary scatter matrices $S_1$ and $S_2$, implying that it actually applies not just in ICA, but in the wider framework of ICS, including also the sufficient dimension reduction methodology of \cite{li2007sparse}. For simplicity, we restrict ourselves in the following results to the case $k = 1$, which means that we estimate the first IC (or invariant coordinate, in the general case) only. Given sample estimates $S_{n1}, S_{n2}$ of the two scatter matrices, we let $b_{n, \lambda_{1n}} \in \mathbb{R}^p$ denote the minimizing value of $\beta_1$ (the first column of $B$) in the sample version of the optimization problem~\eqref{eq:main_optimization} when the penalization parameter is $\lambda_1 \equiv \lambda_{n1}$. Similarly, we let $b \in \mathbb{R}^p$ denote the minimizing value of $\beta_1$ in the corresponding population level problem with no penalization. By Theorems \ref{theo:ic_solution_1} and \ref{theo:ic_solution_2}, when $S_1$ and $S_2$ have the independence property, the vector $b$ thus equals the first row of the unmixing matrix $\Omega^{-1}$, up to sign.  As is typical in convergence results related to LASSO-like methods \citep{fu2000asymptotics, chatterjee2011strong}, a necessary condition for the convergence of $b_{n, \lambda_{1n}}$ to $b$ is that the penalty parameter $\lambda_{n1}$ vanishes at an appropriate rate as $n \rightarrow \infty$. Theorem~\ref{theo:main} below is proven in several steps in \ref{sec:proofs} and follows from Theorem \ref{theo:main_result} therein, by selecting $a_n := 1/c_n$.

\begin{theorem}\label{theo:main}
    Let $S_{n1}, S_{n2}$ be two $p \times p$ sample scatter matrices satisfying
    \begin{align*}
        c_n (S_{n1} - S_1) = \mathcal{O}_p(1) \quad \mbox{and} \quad c_n (S_{n2} - S_2) = \mathcal{O}_p(1)
    \end{align*}
     for some sequence $c_n$ satisfying $c_n \rightarrow \infty$ and some positive definite $S_1, S_2$. Let $\lambda_{n1} \rightarrow 0$ be such that $\lambda_{n1} c_n \rightarrow 0$. Then,
    \begin{align*}
        c_n \| s_n b_{n, \lambda_n} - b \| = \mathcal{O}_p(1),
    \end{align*}
    for some sequence $s_n \in \{ -1, 1\}$ of signs.
\end{theorem}

In a typical case, the convergence rate of the scatter estimates would be $c_n = \sqrt{n}$, see e.g. \cite{ilmonen2010characteristics, miettinen2015fourth}, meaning that the ``optimal'' (most sparsity inducing) choice of the penalization parameter is $\lambda_{n1} = n^{-1/2 - \varepsilon}$ for some arbitrarily small $\varepsilon > 0$. By Theorem \ref{theo:main}, this choice then leads to $b_{n, \lambda_n}$ also inheriting the convergence rate $\sqrt{n}$ from the two scatter matrices.

We still give a short summary of the techniques of proof used in showing Theorem \ref{theo:main}: We break the problem down in the tasks of separately controlling the errors between $b_{n, \lambda_n}$ and $ b_{n, 0}$ and between $ b_{n, 0}$ and $b$. The latter task does not depend on the sparsity parameter and is based on manipulating the matrix form of the generalized eigendecomposition of $S_{n2}$ w.r.t. $S_{n1}$ to establish the convergence in a specific coordinate system and then extending the result to the general case via affine equivariance (Lemma \ref{lemma:from_II_to_III}). The former task (error between $b_{n, \lambda_n}$ and $ b_{n, 0}$) is more tedious, and involves obtaining a sequence of lower bounds for the difference of the two objective functions and showing that, as soon as $\lambda_{n1}$ vanihses at a suitable rate, the corresponding minimizers are within an $\varepsilon$-neighbourhood of each other with increasing probability. The combining of the two results then yields Theorem \ref{theo:main}.

\subsection{Robustness of SICS}\label{sec:robustness_of_sics}

\minorrevision{
In this section, we give a result showing that, when using robust scatter matrices, the SICS algorithm gives robust results. It is clear that when a scatter matrix breaks in sense of breakdown point, algorithms based on it also break. The following Theorem \ref{theo:main_2} states that if neither of the scatter matrices break, also SICS does not break. Therefore, the breakdown point of SICS satisfies $\varepsilon^*(\text{SICS}) = \min(\varepsilon^*(S_1), \varepsilon^*(S_2))$. The proof of the theorem is given in \ref{sec:proofs_2}. Let us denote in the following the set of $p \times p$ positive definite matrices by $\mathcal{S}_+^p$ and the $k$th largest eigenvalue of a matrix $S$ by $\phi_k(S)$.
}
\minorrevision{

\begin{theorem}\label{theo:main_2}
    Let $f(A, B; S_1, S_2) = \| S_1^{-\frac{1}{2}} S_2^{\frac{1}{2}} - A B^\top S_2^{\frac{1}{2}} \|_F^2 + \lambda\sum_{m=1}^k \| \beta_m \|_1$ be the SICS objective function with $A,B = (\beta_1,\dots,\beta_k) \in \mathbb{R}^{p \times k}$, $A^\top A~=~I_k$ and fixed $\lambda \geq 0$. Let $(S_{10}, S_{20}) := (S_1(X), S_2(X))$ be the scatters $S_1, S_2$ evaluated for some fixed sample $X$. Fix $\varepsilon > 0$ and let
    \begin{align*}
        \mathcal{G} = \{ (S_1(X_\varepsilon), S_2(X_\varepsilon) ) \mid X_{\varepsilon} \mbox{ is an } \varepsilon \mbox{-contamination of } X \} \subset \mathcal{S}_+^p \times \mathcal{S}_+^p.
    \end{align*}
    Assume that the following hold:
    \begin{itemize}
        \item[(i)] We have $\inf_{(S_1, S_2) \in \mathcal{G}} \phi_p(S_2) > 0$. 
        \item[(ii)] We have $\sup_{(S_1, S_2) \in \mathcal{G}} \| S_1^{-\frac{1}{2}} - S_{10}^{-\frac{1}{2}} \|_F < \infty$ and  $\sup_{(S_1, S_2) \in \mathcal{G}} \| S_2 - S_{20} \|_F < \infty$.
    \end{itemize}
    Then, letting $(A(S_1, S_2), B(S_1, S_2))$ denote any minimizer of $(A, B) \mapsto f(A, B; S_1, S_2)$, we have
    \begin{align*}
        \sup_{(S_1, S_2) \in \mathcal{G}} \| B(S_1, S_2) - B(S_{10}, S_{20}) \|_F < \infty.
    \end{align*}
\end{theorem}
}

\minorrevision{Theorem \ref{theo:main_2} is stated in terms of samples $X$, but it extends trivially to general distributions and their contaminations. The conditions in assumption $(ii)$ in Theorem \ref{theo:main_2} are asymmetric in $S_1$ and $S_2$ due to the form of the objective function. However, Theorem \ref{theorem2_res4} in \ref{sec:proofs_2} gives that the first condition in assumption $(ii)$ holds if we assume $\sup_{(S_1, S_2) \in \mathcal{G}} \| S_1 - S_{10} \|_F < \infty$ and $\inf_{(S_1, S_2) \in \mathcal{G}} \phi_p(S_1) > 0$, leading to sufficient conditions that are symmetric in $S_1, S_2$. Finally, we note that by taking $\lambda = 0$, Theorem \ref{theo:main_2} covers also the robustness of regular ICS.}


\subsection{Computational aspects}\label{sec:computation}

\minorrevision{The computational burden of SICS comes from computing the scatter matrices $S_1$ and $S_2$ and running Algorithm~\ref{alg:SICS} for them, the latter comprising of multiple steps.} The complexity of the scatter computation depends greatly on the particular choices of scatters. For example, the classical FOBI-estimate uses the moment-based matrices $S_1 = \mathrm{Cov}(x)$ and $S_2 = S_{\mathrm{FOBI}}$ which are very fast to compute even for high $n$ and $p$, whereas the computation of symmetrized robust scatter matrices scales rather badly with both the sample size and data dimension. \minorrevision{The runtime of Algorithm~\ref{alg:SICS} is divided to (a) computing the standard ICS-solution, (b) updating the $\beta_j$ using \texttt{solvebeta} \cite{Relasticnet}, a LASSO-solver based on Cholesky decompositions and efficient tracking of inactive and active variables, and (c) matrix factorizations of low-rank matrices.} To demonstrate the computational burden of SICS, we simulated data from the model described later in Section \ref{sec:simulations} with the choices $n = 200, 400, 600$ and $p = 10, 20, \ldots, 100$ and measured the average time (over 20 replications) it takes to estimate $k = 5$ sparse independent components with the sparsity parameter $r = 7$. \minorrevision{The full timing was divided into four parts: scatter computation, lines 1--3 of Algorithm \ref{alg:SICS} (standard ICS), lines 5--9 (LASSO-update) and lines 10--13 (decompositions).} We considered two versions of SICS, non-robust (FOBI) and robust (symmetrized Huber and $t_1$-MLE). \minorrevision{The timings were measured on a 14'' MacBook Pro M3 with 16 GB memory.}

\begin{figure}
    \centering
    \includegraphics[width=0.85\linewidth]{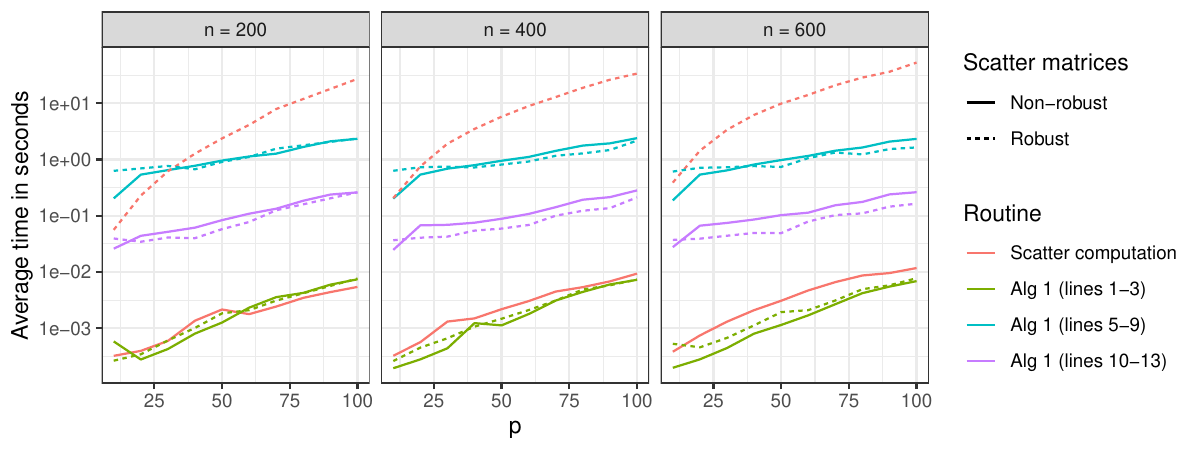}
    \caption{Average computation time of SICS as a function of sample size $n$ and data dimension $p$, divided into the computation of the scatters and the different parts of Algorithm \ref{alg:SICS}. The $y$-axis has a logarithmic scale.}
    \label{fig:timing_breakdown}
\end{figure}

\minorrevision{The results are depicted in Figure \ref{fig:timing_breakdown} and show that the main bottleneck of SICS is the computation of robust scatters which scales superlinearly in both $n$ and $p$ (red dashed line). This is caused by the symmetrization operation which is needed for the independence property and inflates the effective sample size to $n^2$, see Section \ref{sec:ica_and_scatter}. In comparison, computing the FOBI-estimate is extremely fast for all $n, p$. As Algorithm \ref{alg:SICS} consists of optimization routines and matrix decompositions for the pre-computed $p \times p$ scatters $S_1, S_2$, it is natural that its computation time depends very little on $n$ and the choice of scatters. Of the three components of Algorithm \ref{alg:SICS}, the computation of the standard ICS-solution (lines 1-3) is the fastest part. This is as expectable since the other two components require iteration (the repeat-until loop). Computing the matrix decomposition (lines 10-13) is an order of magnitude faster than the LASSO-update (5-9). This is because the matrices in question are low-rank, having sizes $p \times k$ and $k \times k$ with $k = 5$ here, leading to efficient computation of the decompositions.

A possible way to alleviate the computation cost of the robust scatters would be to use in the symmetrization only a subset of the pairwise distances \cite{dumbgen2024approximating}, trading estimation accuracy for computational speed. Similarly, to manage large values of $p$, one could run SICS repeatedly with several random small subsets of variables to gain a preliminary idea on which features are the most important, pruning the rest, and then running SICS on the resulting data with smaller $p$ to obtain the final estimate. Since the data we study in this work are still manageable in size, we did not resort to these shortcuts.}

\section{Simulations}\label{sec:simulations}

\revision{The following subsections describe the results of four simulations studies. These simulations focus on, (1) a general comparison between robust/non-robust and sparse/non-sparse ICS, (2) the impact of the level of sparsity on the estimate, (3) behavior of the method under large dimensionality, and (4) comparison between SICS and other methods of sparse ICA, respectively.}

\subsection{Simulation study \#1: Benefits of robustness and sparsity}

We investigated the performance of the proposed method SICS with simulations. The task was to estimate the sparse unmixing matrix $\Omega^{-1}$ for a data matrix $X = (x_1, \ldots, x_n)^\top$ generated from the IC model $X = Z \Omega^\top$ where $Z = (z_1, \ldots, z_n)^\top$. In every iteration, we first generated a random unmixing matrix $\Omega^{-1} \in \mathbb{R}^{p{\times}p}$, where every row has a preset number $q$ of non-zero elements coming from the Uniform$([-3,-1] \cup [1,3])$-distribution. Then, we generated a random source matrix $Z \in \mathbb{R}^{n{\times}p}$ with the elements of the first column coming from the standardized Laplace distribution, the second column coming from standardized uniform distribution and rest coming from the standard normal distribution. Last, we generated the data matrix $X = Z \Omega^\top$ using $Z$ and $\Omega$. We contaminated the data matrix $X$ by replacing $5\%$ of the rows with observations coming from $N(0,3)$. All simulations were also run similarly to the non-contaminated datasets to evaluate the effect of the contamination to the estimators. We used different variations of proposed SICS algorithm to estimate the coefficient vector $\beta_1$ of the first independent component and calculated its absolute distance to the true first row of the matrix $\Omega$. Because the sign of the coefficient vector is arbitrary, distance to vector with flipped signs was used if it was smaller. Then we took the median of this over $N$ iterations.

The methods to compare were SICS with four different numbers of estimated non-zero coefficients $r$: full $p$ (non-sparse estimate), true $q$ (oracle estimate), $q+3$ and ``hard $q+3$'' where all $p$ coefficients were first estimated in a non-sparse way and then the $p-q-3$ with the smallest absolute values were set to zero (making it a hard thresholding estimate). All of those were run with a non-robust pair of scatter matrices, the covariance matrix and the FOBI matrix, and a robust pair, symmetrized $t$-distribution based ($\nu = 1$) and symmetrized Huber's $M$-estimator (where $c$ is chosen so that $\mathbb{P}(\chi^2_p \leq c^2/2)=0.9$). 

First, we chose the number of variables to be $p=15$, of which non-zero $q=7$. We varied the number of observations $n=500,1000,1500,2000,2500$ and took $N=1000$ repetitions for each $n$. The results are presented in Figure \ref{fig:error_vs_n}. When the dataset is contaminated, methods based on robust scatter matrices perform clearly better than non-robust, as one could expect. SICS with $r=q+3$ performs the best after $n=1500$. With non-contaminated data, the difference between robust and non-robust is not as clear as with contaminated data, but, overall, robust variants still seem to perform better than the non-robust. With high $n$, we see that all robust methods except the choice $r=q$ perform about the same. Overall we can conclude that higher sample size $n$ increases performance, as expected, and that the robust variants outmatch the non-robust ones.

Then, we chose the number of variables to be $p=15$ and the number of observations to be $n=1500$. We varied the number of non-zero coefficients $q=2,3,5,8,11,15$ for again $N=1000$ repetitions for each $q$. The results are presented in Figure \ref{fig:error_vs_q}, where the $y$-axis is standardized by dividing by $q$ which is proportional to the $l_1$-norm of a constant vector with $q$ non-zero coefficients. We can again see that with contaminated data robust variants work best. When the number of non-zero coefficients is small, sparse variants work better, as one could guess. We can also see again that these differences, and similarly the effect of the robustness, are not so clear in the case of the non-contaminated data.

\begin{figure}
    \centering
    \includegraphics[width=0.8\linewidth]{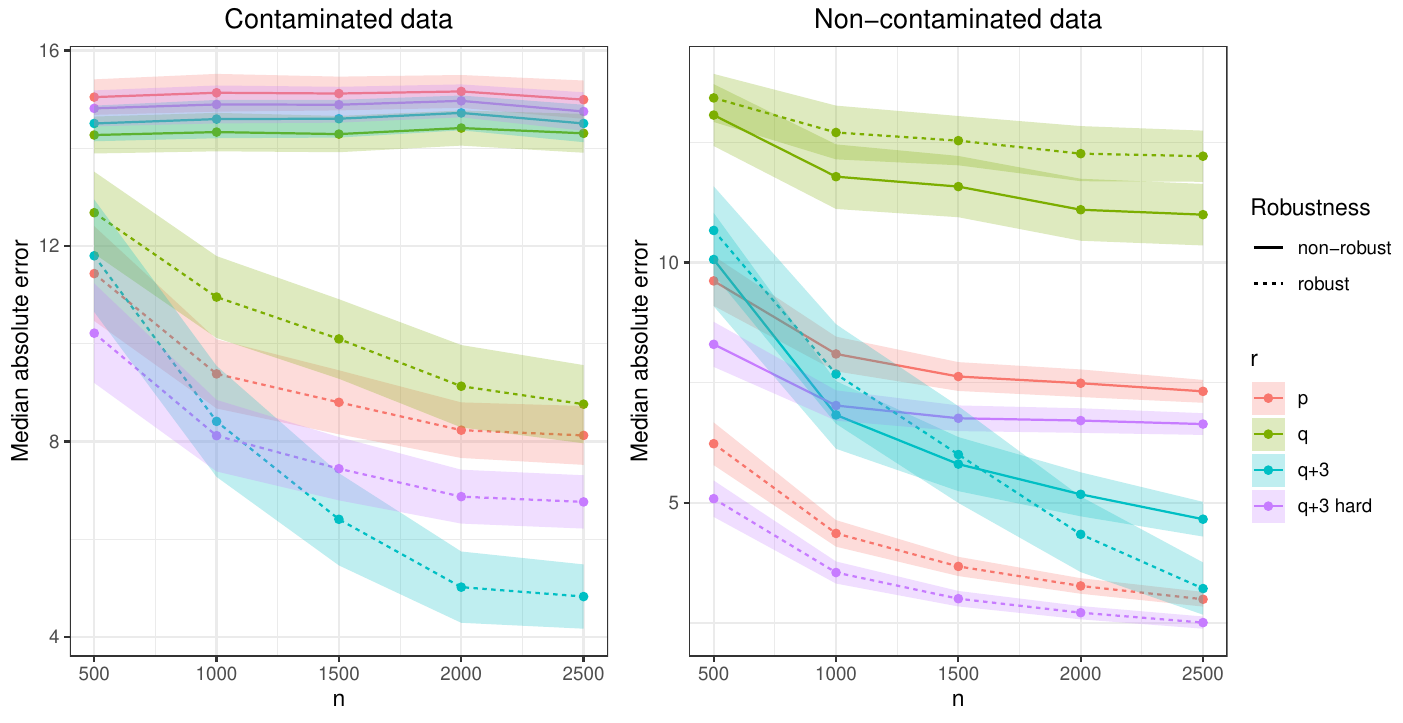}
    \caption{Median absolute error by sample size for different methods. The error ribbon has width $0.2\times$MAD}
    \label{fig:error_vs_n}
\end{figure}
\begin{figure}
    \centering
    \includegraphics[width=0.8\linewidth]{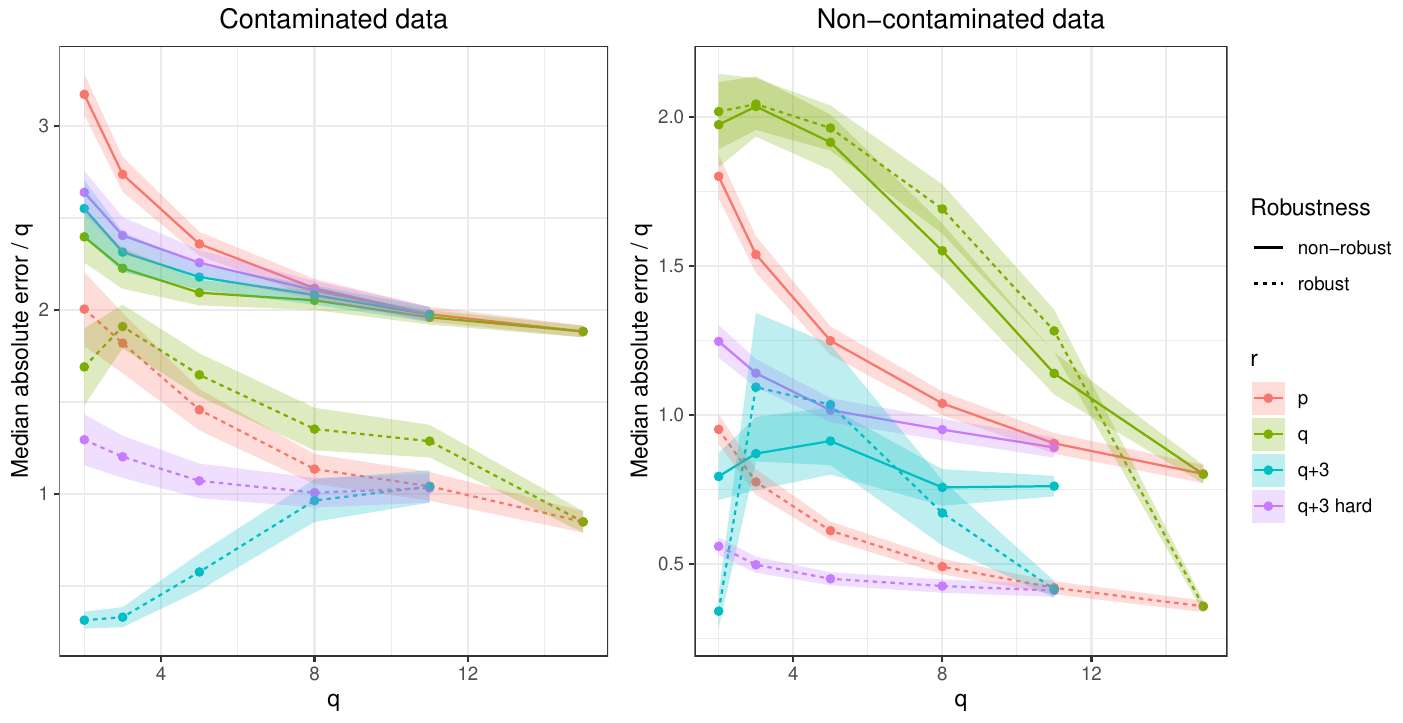}
    \caption{Median absolute error by number of non-zero coefficients for different methods. The error ribbon has width $0.2\times$MAD. Error is scaled by dividing by $q$.}
    \label{fig:error_vs_q}
\end{figure}

\subsection{Simulation study \#2: Choosing the regularisation parameter}

The purpose of our second simulation is to investigate the optimal choice of $\varnum$ (sparsity level of the estimate) for a given $q$ (the true sparsity level). Naturally, it is clear that $\varnum \geq q$ is the minimal requirement for successful estimation but, based on the previous simulations, having $\varnum > q$ is actually preferable and we next study what amount of ``overestimation'' is optimal for finite samples. We generated data from the same model as earlier, with $n = 1000$, a contamination level equal to $5\%$, dimensions $p = 10, 15, 20$ and $q = \lfloor \alpha p \rfloor$ for $\alpha = 0.1, \ldots, 1.0$. We then estimated the first IC using the same robust pair of scatter matrices as in the earlier simulation, separately for each $\varnum = 1, \ldots, p$, and recorded the value of $\varnum$, denoted as $\varnum_{\mathrm{opt}}$ yielding the smallest median absolute error. The simulation was repeated 200 times for each combination of $p, \alpha$, and in Figure \ref{fig:opt_varnum} we give the average values of $\varnum_{\mathrm{opt}}/p$ over the replicates.

\begin{figure}
    \centering
    \includegraphics[width=0.6\linewidth]{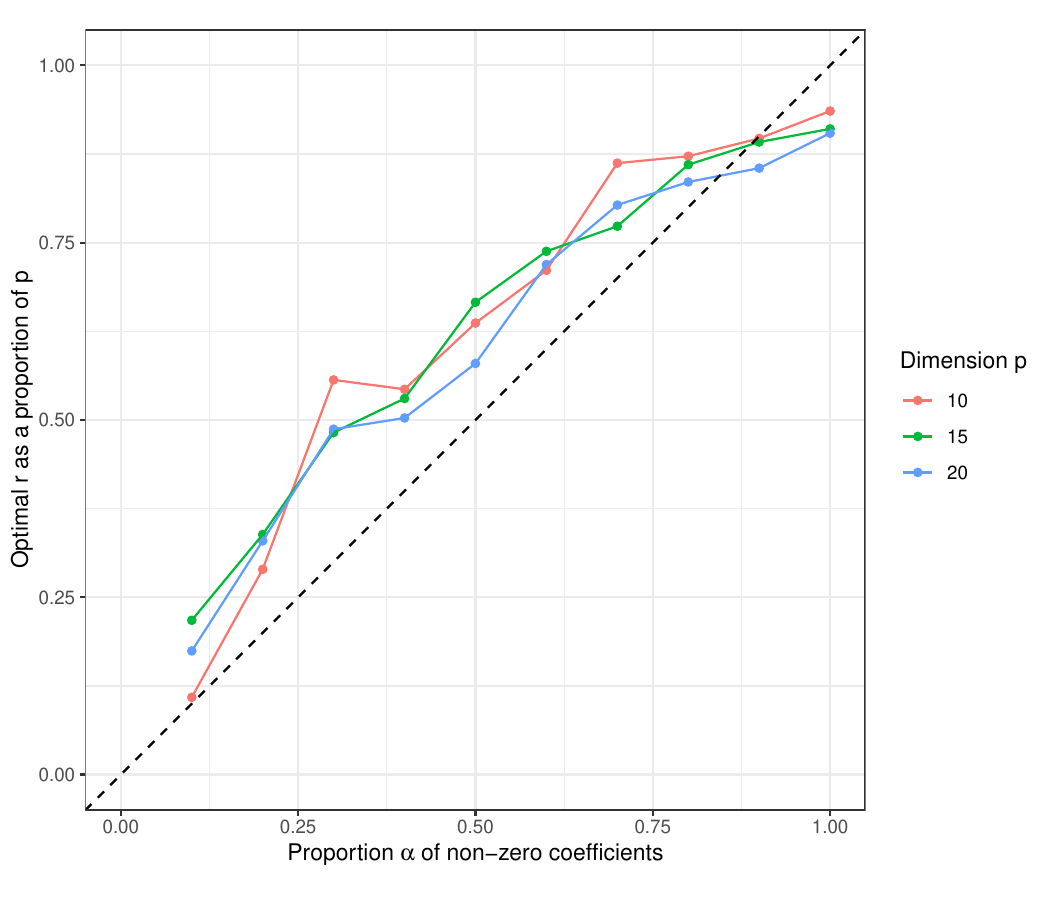}
    \caption{The plot shows the average optimal values of $\varnum$ (as proportions of $p$) as a function of $\alpha$, the true proportion of non-zero coefficients in the first IC. The different lines correspond to different dimensions $p$.}
    \label{fig:opt_varnum}
\end{figure}

The results reveal that the optimal choice of $\varnum$ (as a proportion of $p$) depends very little on the dimension $p$. Moreover, the amount of optimal overestimation depends on the true sparsity in the following manner: if there are only a few true non-zero coefficients, then greater overestimation is preferable, and vice versa. \revision{See Sections \ref{sec:selection_tool}-\ref{sec:discussion} for more discussion on the choice of $\varnum$.}

\subsection{Simulation study \#3: Behavior under increasing $p$}\label{sec:simu_3}

\revision{In our next simulation study, we investigate how SICS behaves when the dimensionality $p$ of the data is increased but the number of non-zero loadings stays constant. The first three independent components are generated from $\mathrm{Gamma}(1, 1), \mathrm{Gamma}(2, 1), \mathrm{Gamma}(3, 1)$ distributions, respectively, and the remaining $p - 3$ are Gaussian. The mixing matrix is taken to be block diagonal such that its top left $3 \times 3$ submatrix is the symmetric square root of the matrix $0.5 I_3 + 0.5 J_3$ where $J_3$ denotes the $3 \times 3$ matrix full of ones. The remaining diagonal elements are set to ones. We estimate the first IC loading vector with regular FOBI (ICS) and with sparse version of FOBI produced by our Algorithm \ref{alg:SICS} (SICS). For the latter we use the sparsity level $r = 10$. This choice mimics a practical situation where we do not know the true sparsity level and overshoot to avoid underestimating it. We use $n = 1000, 2000$ and consider $p = 20, 40, \ldots, 200$, and compute the average estimation error (Euclidean distance between the true and estimated loading vector) over $400$ replicates.}

\revision{The results are shown as a function of $p$ in Figure \ref{fig:dimension_fobi}. We first note that the scenario is quite challenging as the number of data points is kept fixed while at the same time we introduce more and more noise to the model. Thus, it is expected that the error rate is increasing in $p$. However, we observe that, despite the wrong choice of $r$, this increase is for SICS much milder than for ICS, showing that the larger the dimension, the more we benefit from sparse estimation. A prospective future task is to study the theoretical growth rate of the error in scenarios where $n, p_n \rightarrow \infty$. For example, for regular LASSO one can achieve, under specific conditions, convergence rate of the order $\{ (\ln p)/n \}^{1/2}$, see \cite[Section 11]{hastie2015statistical}. Finally, interpreting the results from an efficiency viewpoint, we observe from Figure \ref{fig:dimension_fobi}, e.g., that SICS needs only $n = 1000$ observations to achieve equal or better performance than ICS with $n = 2000$.}

\begin{figure}
    \centering
    \includegraphics[width=0.5\linewidth]{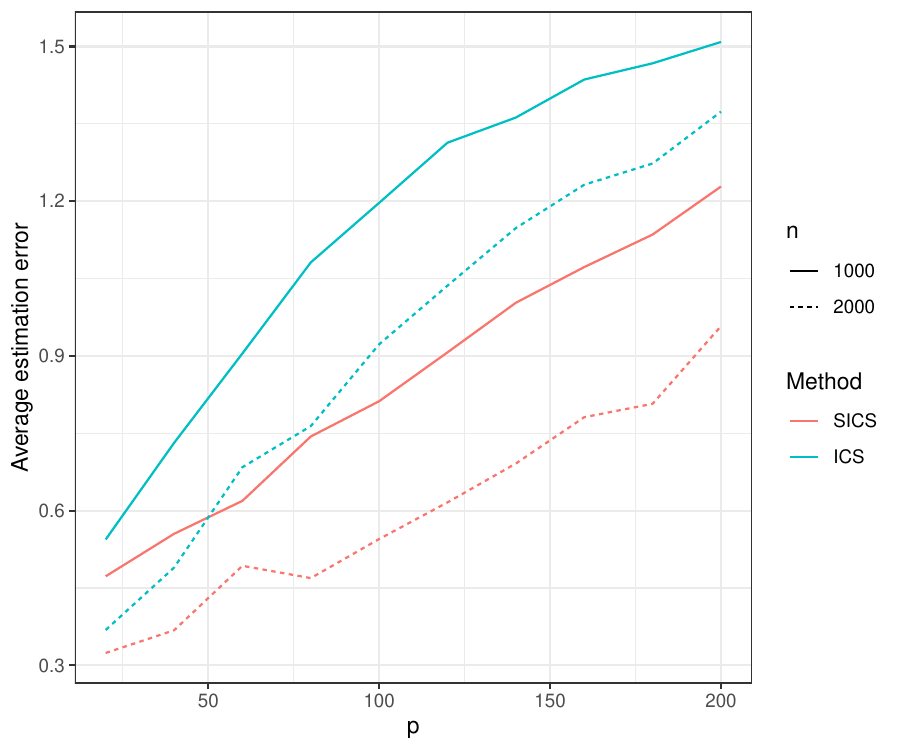}
    \caption{The average estimation error as a function of $p$ in Simulation study \#3}
    \label{fig:dimension_fobi}
\end{figure}


\subsection{Simulation study \#4: Comparison with competing methods}\label{sec:simu_4}

\revision{Next, we compare SICS with some other approaches to doing sparse ICA. A simple way of obtaining a sparse version from the implementations of the well-known ICA methods, JADE and FastICA (in R-packages \texttt{JADE} \citep{RJADE} and \texttt{fICA} \citep{RfICA}), is to use thresholding to put a subset of the coefficients to zero. Let $r$ be the parameter giving the estimated number of non-zero coefficients. In \textit{hard thresholding} $p-r$ smallest coefficients are set to zero after estimation. In \textit{soft thresholding}, in addition to that, all other coefficients are shrunk by the absolute value of the largest coefficient among those set to zero. In addition to these four competitors, we also include a sparse version of FOBI achieved through taking the sparse generalized eigendecomposition of covariance matrix and FOBI-matrix using the implementation of RIFLE \citep{tan2018sparse} in \citep{Rrifle}, with the tuning parameter values $K=1$ and $\lambda = 2\sqrt{\ln p /n}$ in \texttt{initial.convex}.

In this simulation, the basic setup is identical to the first one: We generate a random mixing matrix and then a random data matrix based on that with $p=15$ and the number of non-zero coefficients $q=7$. Then we estimate the first independent component and measure the median absolute error of the loadings. Since none of the competitors is robust, we do not contaminate the data and use non-robust scatter matrices (FOBI) also in SICS. As JADE and fastICA have problems converging when the sample size is in the range we are using, we do two arrangements: First, if JADE does not converge, we ignore that particular instance for JADE. Second, we run fastICA three times for 400 iterations (if it does not converge) and take the best result. Thus, the results for JADE are in some sense presented in too optimistic way as datasets where JADE would not work are not included.

\begin{figure}
    \centering
    \includegraphics[width=0.75\linewidth]{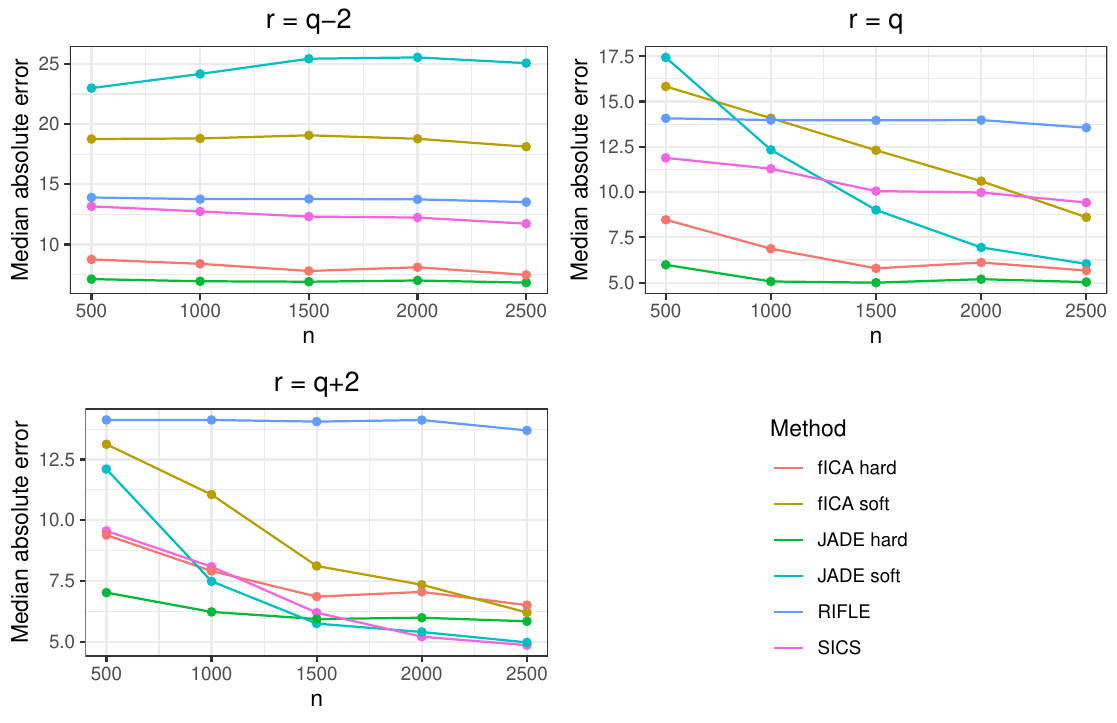}
    \caption{The results of a comparison between SICS and competitors}
    \label{fig:comparison}
\end{figure}

The results over 500 replicates are shown in Figure \ref{fig:comparison} where we can see that the smallest average error is around~5. Multiple methods can reach that, including SICS when $r=q+2$ and $n \geq 2000$. The RIFLE-based FOBI does not seem to perform well but closer inspection shows that it often locates the zero coefficients right, but fails to correctly estimate the non-zero coefficients or their sign. JADE with hard thresholding is in general the most accurate method, which is well-known in non-sparse ICA \cite{miettinen2015fourth}, but it is less reliable in practice and failed to converge in 1571 out of 7500 replicates, skewing the results in Figure \ref{fig:comparison} in its favor. Moreover, the hard thresholding is very heuristic and does not have any theoretical convergence guarantees, unlike our method. Choosing too small $r$, as in $r=q-2$, leads to significantly larger error for all methods. This is larger for soft versus hard thresholding because in soft thresholding putting a non-zero coefficient to zero also reduces all other coefficients without a real reason.

\minorrevision{The median computation times (over 80 repetitions) are presented in Table \ref{tab:comp_times}. The computations were run on Dell Latitude 5520 15,6" i5 with 8GB memory. The timings correspond to case $r=q$, but the results are similar for other values of $r$. It seems that for all method except fICA, the computation times are not dependent on the sample size. This is because the other methods operate directly on $p \times p$ matrices that are computed only once, whereas fICA requires repeated evaluations of objective function gradients separately for each observation.}

\begin{table}
{\footnotesize
    \centering
    \begin{tabular}{r|c|c|c|c|c}
        $n$ & 500 & 1000 & 1500 & 2000 & 2500 \\ \hline
        SICS & 0.448 (0.152) & 0.495 (0.123) & 0.501 (0.144) & 0.488 (0.198) & 0.505 (0.229) \\
        fICA & 0.780 (0.047) & 1.309 (0.095) & 1.952 (0.267) & 2.440 (0.187) & 3.042 (0.312) \\
        JADE & 0.045 (0.012) & 0.072 (0.012) & 0.092 (0.011) & 0.117 (0.013) & 0.143 (0.019) \\
        RIFLE & 0.089 (0.116) & 0.059 (0.078) & 0.069 (0.092) & 0.072 (0.094) & 0.063 (0.083) \\
    \end{tabular}
    \caption{Median running times in seconds (with MAD in parenthesis) for different methods and sample sizes in simulation \#4.}
    \label{tab:comp_times}
    }
\end{table}

}

\section{Robust causal discovery}\label{sec:real_data}

We next demonstrate how the proposed method can be used in robust construction of causal graphs. Given a dataset with $p$ variables, the objective in \textit{causal discovery} is to form a directed acyclic graph whose nodes are the variables and whose edges represent causal relations between the variables \citep{pearl2009causality}. The acyclicity and directedness of the graph then ensures that ``effects cannot precede causes''. \cite{shimizu2006linear} showed the remarkable fact that ICA can be used for non-Gaussian linear causal discovery. Essentially, this is because a linear non-Gaussian causal model between the elements of a random $p$-vector $x$ can be written as $x = Bx + \varepsilon$, where $\varepsilon$ represents measurement noise with non-Gaussian independent components and the $p \times p$ matrix $B$ is strictly lower triangular. By rearranging the terms, we observe that this is actually an IC model, $x = (I - B)^{-1} \varepsilon$ and thus estimable with ICA.

In general, the unmixing matrix estimates produced by ICA-methods are not lower triangular (or permutations thereof) and, as such, typical ICA solutions do not correspond to causal graphs. However, this can be forced with the following two steps: (i) By using a sparse ICA method, the estimate of $I - B$ becomes sparser and usually closer to being a permutation of a lower triangular matrix, see \cite{harada2020estimation}. (ii) By permuting the result of a sparse ICA method to be as lower triangular as possible and afterwards pruning the excess elements to zero, e.g., with the efficient and scalable algorithm in \cite{hoyer2006new}. See also \cite{ng2023identifiability} for a second-order ICA method whose identifiability constraints make it naturally suited to causal discovery.

In this experiment, we apply the previous causal discovery approach to the \texttt{diabetes} dataset available in the R-package \texttt{elasticnet} \citep{Relasticnet}. The data consist of the measurements of $10$ baseline covariates (such as age, sex and bmi) and a response variable (disease progression score (DP)), for $n = 442$ patients. To discover causal relationships between this set of $p = 11$ variables, we use our proposed method to construct causal graphs as described above, see \cite{hoyer2006new} for details of the permutation and pruning algorithm. We distinguish two versions of our method: for the first $S_1$ and $S_2$ equal the covariance matrix and the FOBI-matrix (non-robust method) and for the second $S_1$ and $S_2$ equal the symmetrized MLE of $t_1$-distribution and the symmetrized Huber's $M$-estimator with the tuning parameter value $0.9$ (robust method). In both cases, we estimate a full set of $p$ independent components with $\varnum = 7$ non-zero coefficients per component. As each produced causal graph can be seen as a single ``point estimate'', we bootstrap the original data set $1000$ times, and use both methods (non-robust and robust) to estimate the causal graph for each bootstrap replicate. The final, aggregated graphs are then formed by retaining only those directed edges which are present in at least $40\%$ of the bootstrap graphs, separately for both methods. This low percentage was chosen since the produced graphs exhibited quite a lot of variability, which stems from the fact that the sample size $n = 442$ is relatively small from the viewpoint of ICA.

\begin{figure}
    \centering
    \includegraphics[width=0.70\linewidth]{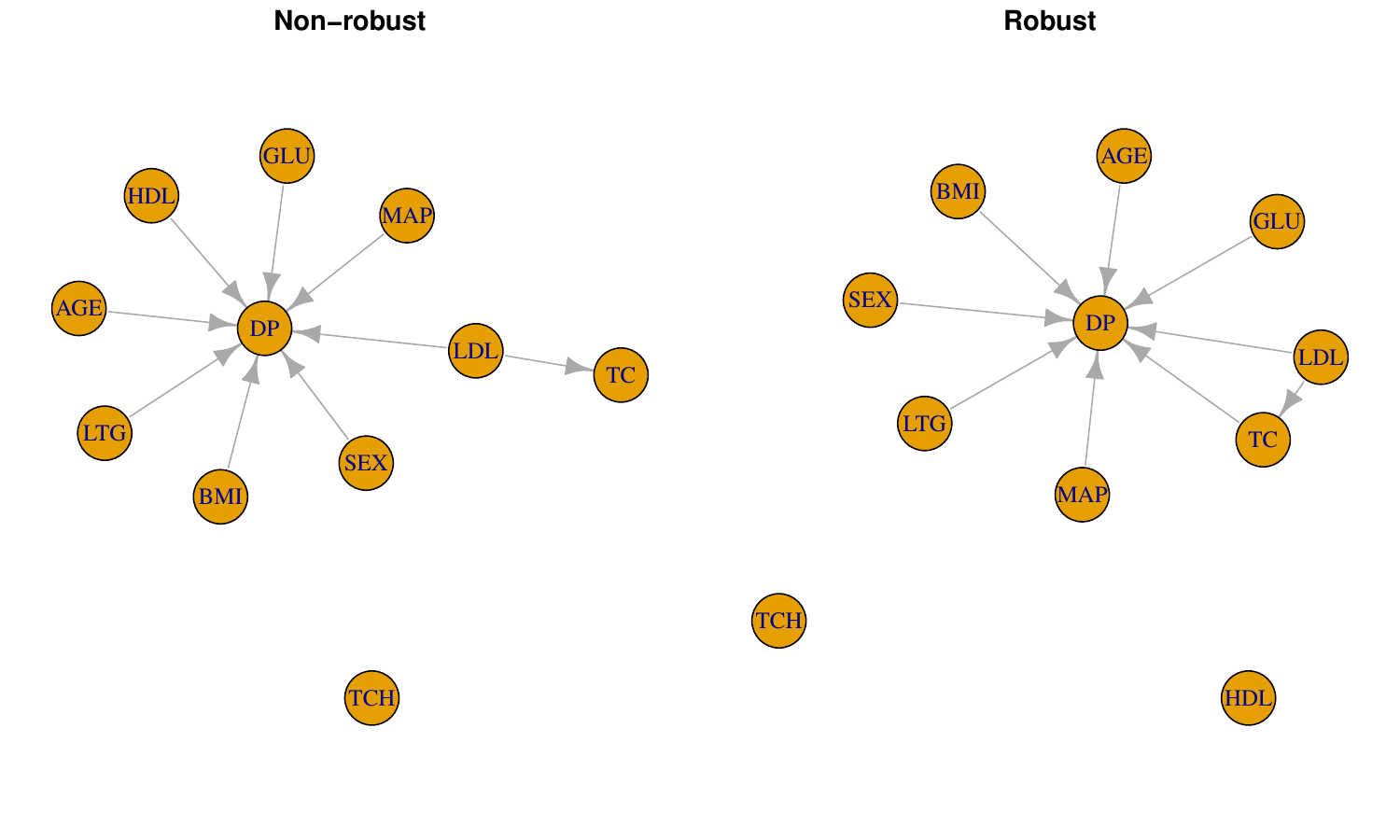}
    \caption{The aggregated causal graphs estimated from the diabetes data by the non-robust and robust method. The variable abbreviations are: DP = disease progression index, AGE = age, SEX = sex, BMI = body mass index, MAP = mean arterial blood pressure, whereas TC, LDL, HDL, TCH, LTG, GLU correspond to specific blood serum measurements.}
    \label{fig:causal_1}
\end{figure}

The obtained causal graphs are shown in Figure \ref{fig:causal_1} and show that both methods more or less agree on the causal structure of the data, indirectly indicating that the data are not likely to contain significant outliers (which likely stems from the fact that the data represent a curated clinical study). The main feature of interest is that most of the 10 explanatory variables are estimated to be causes of the disease progression variable DP. Additionally, a relation between the covariates LDL and TC was discovered. We note that these plots should not be interpreted as the full causal graphs between the 11 variables, but rather as estimates of the set of \textit{strongest causal relationships} between the variables. \minorrevision{We also computed the runtimes (14'' MacBook Pro M3, 16 GB memory) of producing the causal graph both for the non-robust and robust method and these were 11 s and 6 min 44 s, respectively. The difference in the orders of magnitude is in line with the simulation results in Section \ref{sec:computation}.}

We next take these findings as the ground truth, and continue the experiment by creating a contaminated version of the data set, obtained by randomly selecting $10\%$ of the subjects and replacing all their measurements with i.i.d. Gaussian noise with standard deviation $\sigma = 20$. We then applied the same causal discovery approach to this contaminated data set, with the hopes of still being able to find the relevant structure, despite the contamination. The bootstrapped causal graphs estimated from the contaminated data are shown in Figure \ref{fig:causal_2} and reveal that the structure of the non-robust graph has completely changed, with only two of the original edges remaining. That is, the outliers have made it impossible for standard, non-robust ICA to find almost any causal relations. However, the graph from robust and sparse ICA retains most of the connections, in particular the edges between the disease progression and the covariates AGE, SEX, BMI, MAP, TC, LDL and GLU. We thus conclude that robust SICS offers a reliable and outlier-resistant method for linear causal discovery. \minorrevision{The runtimes for the contaminated data were 10 s and 12 min 13 s, for the non-robust and robust approach, respectively. The time difference in the robust method for these and the non-contaminated data (12 min 13 s vs. 6 min 44 s) comes from the computation of the symmetrized Huber's M-estimator: the number of pairs with long distances between them is large in the contaminated data, and when these get truncated to constant size by Huber's weight function \citep{sirkia2007symmetrised}, one obtains a less steep gradient and a larger number of iterations is needed for convergence.}

\begin{figure}
    \centering
    \includegraphics[width=0.70\linewidth]{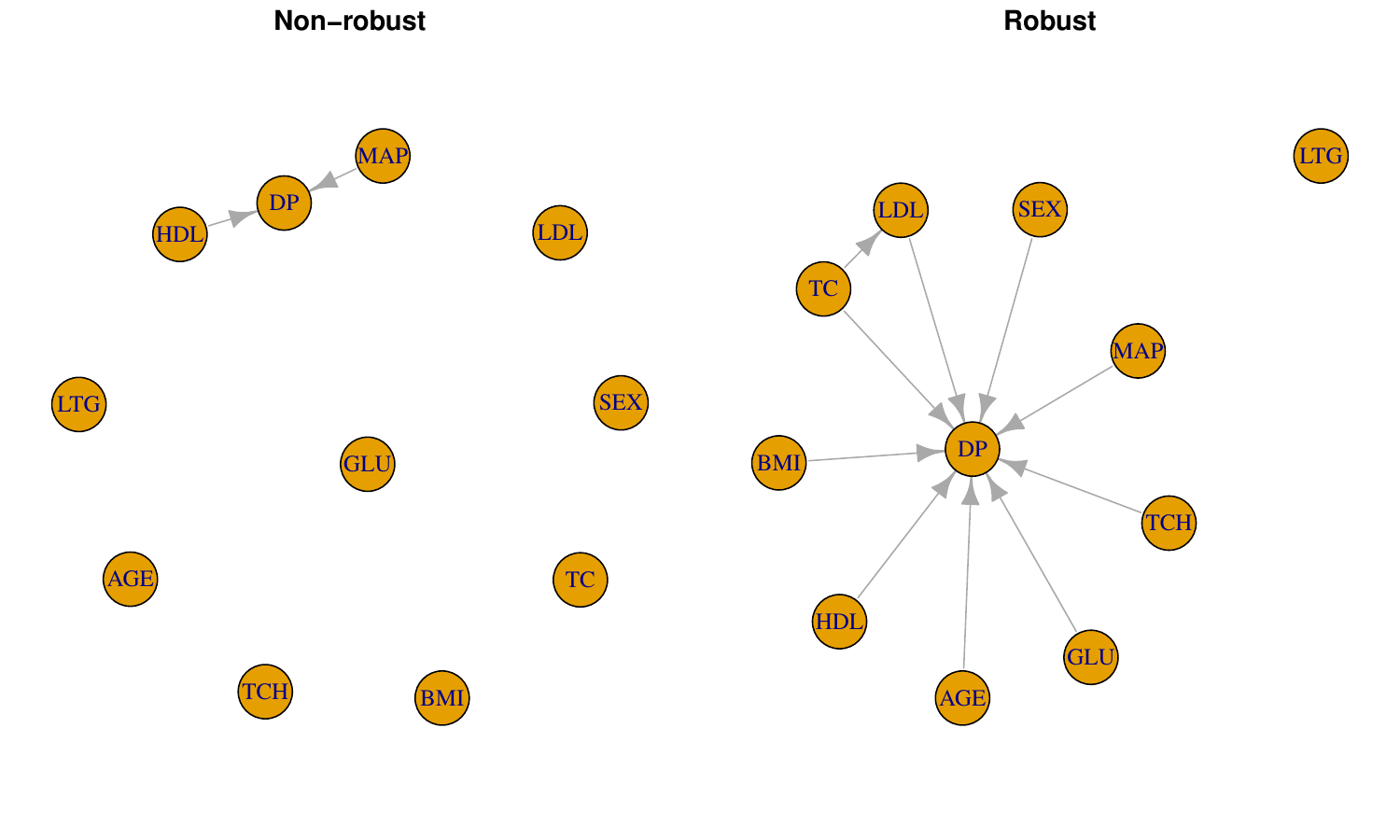}
    \caption{The aggregated causal graphs estimated from the artificially contaminated diabetes data by the non-robust and robust method. See the caption of Figure \ref{fig:causal_1} for the variable abbreviations.}
    \label{fig:causal_2}
\end{figure}

\revision{\section{Visual tool for variable selection}\label{sec:selection_tool}

The selection of $r$ in practice is a non-trivial task and the most typical way, cross-validation, is not possible in our unsupervised setting. Therefore, we next propose a graphical tool similar to the well-known stability selection \cite{meinshausen2010stability} to aid in this. One first generates subsamples of size $\lfloor n/2 \rfloor$ from the data without replacement, calculates the SICS estimates with all possible values of $r = 1, \ldots, p$, and computes the selection probabilities $\hat{\Pi}^r_k$ which measure how often, on average over the replicates, each variable $k$ is included in the model for different $r$. In general, variables that are included in the model with high probability are more important than the ones which are often left out. After having selected a cut-off (see below), one can then fit a regular ICS/ICA-model (to achieve non-biased estimates) with only those variables that are deemed important.

Traditional way to interpret the results and select the important variables is to draw a stability paths plot \citep{meinshausen2010stability}. In our case, we draw a plot which, for each of the $p$ variables, depicts the selection probability as a function of $r$. The straight line from $(1, 1/p)$ to $(p, 1)$ corresponds to the inclusion probability under random guessing, and paths that are clearly above this average path (corresponding to the value $\hat{\Pi}^r = r/p$) correspond to the important variables. Similarly as in \cite{meinshausen2010stability}, we also observed that the distinction between the important and non-important variables is typically visually clear. For a more refined and numeric selection rule, we also propose a strategy in which we calculate for all variables the area of their paths above and below the average path. Now variables for which the area above is greater than the area below the line are stated to be important.

We used this approach for the \texttt{diabetes} data described in Section \ref{sec:real_data}, estimating the first independent component loadings using the covariance matrix and FOBI-matrix as $S_1$ and $S_2$. The resulting stability paths based on 1500 subsamples are shown in Figure \ref{fig:stab_paths}. The variables that are (clearly) above the average line for some $r$ are represented with colored lines and the others with gray lines. Clearly DP, LTG and LDL are the most important variables, and TCH and TC are quite important. The area rule is barely positive for TC and clearly negative for HDL (and the variables with paths drawn in gray). We observe, in particular, that the disease progression index (DP) is almost always taken to be part of the model, which is intuitively clear since the data has been collected with the intention of predicting this variable using the others, meaning that DP can be expected to drive the main latent dependency structure of the data.

\begin{figure}
    \centering
    \includegraphics[width=0.7\linewidth]{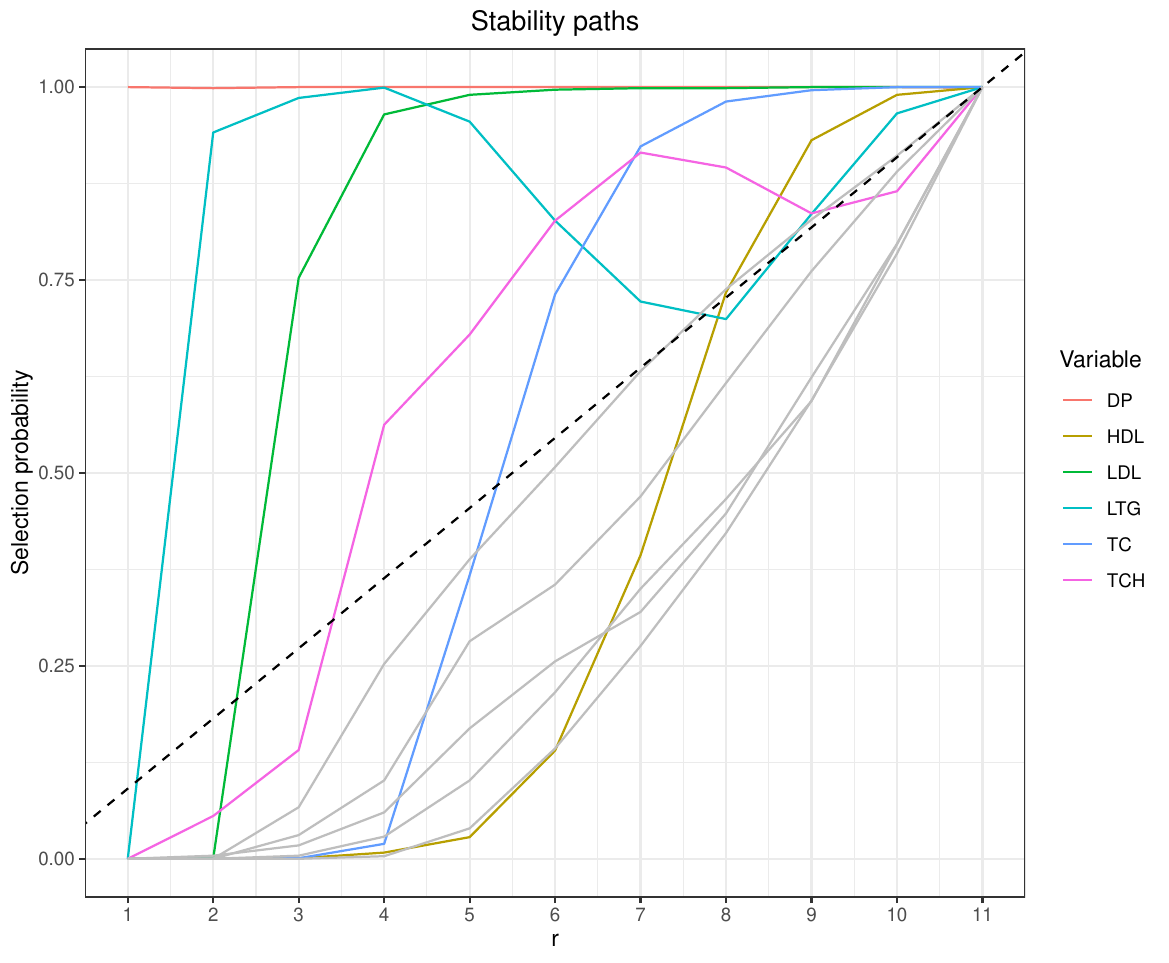}
    \caption{Stability paths for variables in the diabetes data. The gray lines correspond to the variables BMI, AGE, SEX, MAP and GLU.}
    \label{fig:stab_paths}
\end{figure}

}

\section{Discussion}\label{sec:discussion}




We conclude with a discussion about some practical matters and topics for future study. An important choice one has to make when using SICS is selecting the number of non-zero coefficients $\varnum$ (usually chosen to be the same for all estimated components). Based on our simulations, choosing the actual underlying value $q$ is not optimal, but the best option is to choose a value around $10 -20 \%$ greater than $q$. This is because then there is some room for error before a coefficient, which is actually non-zero, is estimated to be zero. It follows that using $\varnum<q$ (for example $\varnum=q-2$) results in clearly bad performance, \revision{as seen in Simulation study \#4.} Choosing $\varnum$ based on the real value $q$ is not usually possible and can be done only in applications where one has some a priori information of the real value $q$. \revision{In other cases, tools such as presented in Section \ref{sec:selection_tool} can be used.}

The consistency result, Theorem \ref{theo:main}, is formulated only for estimating one component ($k=1$). This case is simpler than the general case because one does not have to estimate the rotation matrix $O_k$. We still expect the result to hold for larger $k$ since the matrix $O_k$ consists of eigenvectors of a product of the scatter matrices and the estimated coefficients, which are all either assumed or proven to converge at the rate $c_n$, which is thus expected to be inherited by the estimate of $O_k$ as well.

\revision{On a heuristic level dependencies between the observed variables are highly beneficial in ICA, and ICS in general. This is because, when the variables are more dependent, their joint behavior can, on average, be controlled by a smaller number of latent factors/invariant coordinates. Hence, keeping the model otherwise fixed, introducing stronger dependencies between the variables effectively reduces the number of parameters in the ICS model. As such, besides the high-dimensional regime mentioned in Section \ref{sec:simu_3}, also this scenario (large number of variables all loading on a small set of factors) warrants theoretical study.

Since SICS uses a penalization in the style of classical statistical learning (although not directly in computing the scatter matrices), another natural question is the impact of different scatter matrices on the possible bias-variance tradeoff of the estimator. A natural viewpoint in this context is that using robust matrices decreases the bias, as outliers have a smaller effect, but might increase the variance, as the effective sample size decreases because of downweighting some observations. Some preliminary testing (not shown here) though reveals that the robust version has a bit smaller bias but significantly smaller standard deviation. This indicates that, robustness can help reduce not only the variance but also bias, highlighting the benefits of using robust scatter matrices.
}

\subsection*{Acknowledgments}

\revision{The authors are grateful to the anonymous reviewers whose comments were of great help in improving the manuscript. The authors would like to thank Henri Nyberg for his constructive comments and feedback on the manuscript.}  The work of LH and JV was supported by the Research Council of Finland (grants 347501, 353769, 368494). The authors would like to thank Andreas Artemiou who brought the work by \cite{li2007sparse} to their attention.

\appendix

\section{Comparison table of sparse ICA methods}\label{sec:comparison_table}

    \begin{table}[h]
    {\footnotesize
        \centering
        \begin{tabular}{r|c|c}
            \hline
            Method & Assumptions & Estimation \\ \hline
            SICS & Distinct kurtoses & Least squares + Procrustes rotation \\
            Bayesian Sparse ICA \cite{hyvarinen2002imposing} & Likelihood & ML \\
            SCAD-ICA \citep{zhang2006ica} & Likelihood & ML + natural gradient \\
            Adaptive $L_1$-ICA \citep{zhang2009ica} & Likelihood & ML + adaptive step size \\
            OBS-ICA \citep{zhang2009ica} & Likelihood + technical assumptions & ML + pruning \\
            sgnICA \cite{palsson2014sparse} & Gaussian time-dependent data & ML + EM-algorithm \\
            Sparse Gaussian ICA \cite{abrahamsen2018sparse} & Gaussian data + sparse and generic $\Omega$ & SCIP \\
            SG-ICA \cite{chen2019sparse} & Likelihood + technical assumptions & ML + EM-algorithm \\
            sICA-LiNGAM \cite{harada2020estimation} & Likelihood & ML + ADMM \\
            SparseICA \cite{ng2023identifiability} & Structural variability of $\Omega$ & ML (or matrix decompositions) \\
            \hline
            Method & Theoretical guarantees & Sparsity \\ \hline
            SICS & Cons. & $L_1$ \\
            Bayesian Sparse ICA \cite{hyvarinen2002imposing} & - & $L_1$ \\
            SCAD-ICA \citep{zhang2006ica} & Oracle property & SCAD \\
            Adaptive $L_1$-ICA \citep{zhang2009ica} & Model selection cons. & Adaptive $L_1$ \\
            OBS-ICA \citep{zhang2009ica} & Equivalence with information criteria & OBS \\
            sgnICA \cite{palsson2014sparse} & - & $L_0$ \\
            Sparse Gaussian ICA \cite{abrahamsen2018sparse} & High-dimensional error rate & Thresholding \\
            SG-ICA \cite{chen2019sparse} & Cons. + Selection cons. + AN & $L_1$ + group $L_1$ \\
            sICA-LiNGAM \cite{harada2020estimation} & - & Adaptive $L_1$ \\
            SparseICA \cite{ng2023identifiability} & Cons. & $L_0$
        \end{tabular}
        \caption{Comparison table between sparse ICA methods. ADMM = Alternating Direction Method of Multipliers, AN = Asymptotic Normality, Cons. = Consistency, ML = Maximum Likelihood, OBS = Optimal Brain Surgeon, SCAD = Smoothly Clipped Absolute Deviation, SCIP = Single Column Identification Procedure.}
        \label{tab:sICA_comparison}
        }
    \end{table}

\revision{In Table \ref{tab:sICA_comparison}, we present a comparison of SICS and the non-robust sparse ICA methods described in Section \ref{sec:introduction}. In the assumptions-column, the term ``Likelihood'' means that the method is based on maximum likelihood, which typically means that the theoretical guarantees of the method hold only if the likelihood has been specified approximately correctly. All listed methods require existence of suitable moments (usually 2nd or higher), so these are not mentioned in the table. Moreover, most methods require the non-Gaussianity of sources. The exceptions to this are sgnICA \cite{palsson2014sparse}, Sparse Gaussian ICA \cite{abrahamsen2018sparse} and SparseICA \cite{ng2023identifiability}.}

\section{Proof of Theorem \ref{theo:main}}\label{sec:proofs}

In this section we prove Theorem \ref{theo:main} in several steps. For clarity (to avoid the use of multiple subscripts), the notation in the proofs differs slightly from the main text.

We assume that $M_n, G_n$ are symmetric $p \times p$ matrices such that $M_n \rightarrow_p M, \quad G_n \rightarrow_p G$, as $n \rightarrow \infty$, where $M, G$ are symmetric and positive definite $p \times p$ matrices. The first, second and last eigenvalues of $G_n^{-1/2} M_n G_n^{-1/2}$ are denoted in the following by $\rho_n, \psi_n, \Omega_n$, respectively. The population counterparts of these are denoted by $\rho, \psi, \Omega$ and we assume that $\rho > \psi$ and $\Omega > 0$. The population and sample eigenvectors corresponding to the leading eigenvalues $\rho_n, \rho$ are denoted by $u_n, u \in \mathbb{R}^p$, respectively. Finally, we denote the smallest eigenvalues of $G_n$ and $G$ as $\alpha_n$ and $\alpha > 0$, respectively. 

The following lemma and corollary follow from the results of \cite{li2007sparse} and we thus omit their proofs.

\begin{lemma}\label{lemma:alternative_form}
    For $a, b \in \mathbb{R}^p$, $\| a \| = 1$, we have $\sum_{i = 1}^p \| G^{-1/2} r_i - a b^\top r_i\|^2 = \mathrm{tr}(G^{-1/2} M G^{-1/2}) - 2 b^\top M G^{-1/2} a + b^\top M b$, where $(r_1, \ldots, r_p) = M^{1/2}$.
\end{lemma}

\begin{corollary}\label{cor:gee_via_optimization}
    The minimizer $(\alpha, \beta)$ of $(a, b) \mapsto \sum_{i = 1}^p \| G^{-1/2} r_i - a b^\top r_i\|^2$ over $a, b \in \mathbb{R}^p$, $\| a \| = 1$ is $\alpha = u$ and $\beta = G^{-1/2} u$.
\end{corollary}

We next equip the objective function in Corollary \ref{cor:gee_via_optimization} with an $\ell_1$-penalty for the parameter $b$. This leads to the following three versions of the optimization problem. The first one is a sample problem with the penalty parameter (sequence) $\lambda_n$, the second one is the non-penalized sample problem, and the third one is simply the population-level problem from Corollary \ref{cor:gee_via_optimization}.

\begin{itemize}
    \item \textbf{Version I:} $(a_{n, \lambda_n}, b_{n, \lambda_n}) = \mathrm{argmin} f_{n, \lambda_n}(a, b)$, $a, b \in \mathbb{R}^p, \| a \| = 1$, 
where $f_{n, \lambda_n}(a, b) = -2 b^\top M_n G_n^{-1/2} a + b^\top M_n b + \lambda_n \| b \|_1$.
\item \textbf{Version II:} $(a_{n, 0}, b_{n, 0}) = \mathrm{argmin} f_{n, 0}(a, b)$, $ a, b \in \mathbb{R}^p, \| a \| = 1$,
where $f_{n, 0}$ is as in Version I above.
\item \textbf{Version III:} $(\alpha, \beta) = \mathrm{argmin} f(a, b)$, $ a, b \in \mathbb{R}^p, \| a \| = 1$, where $f(a, b) = -2 b^\top M G^{-1/2} a + b^\top M b$.
\end{itemize}

Our objective is to show that, under suitable assumptions, $(a_{n, \lambda_n}, b_{n, \lambda_n})$ converges in probability to $(\alpha, \beta)$. I.e., that the $\ell_1$-penalized problem gives consistent solutions. We do this by first showing that the solutions of Versions I and II are close and then doing the same for Versions II and III.

For a vector $m \in \mathbb{R}^p$, we define $a^*_{n, \lambda_n, m}$ to be the minimizer of the objective function $a \mapsto f_{n, \lambda_n}(a, b_{n, 0} + m)$ over $a \in \mathbb{R}^p, \| a \| = 1$. As the objective function is symmetric in the sense that $f_{n, \lambda_n}(a, b) = f_{n, \lambda_n}(-a, -b)$, it is sufficient to restrict our attention to vectors $m \in \mathbb{R}^p$ such that $b_{n, 0}^\top G_n (b_{n, 0} + m) \geq 0$ (corresponding to a certain half-space of $\mathbb{R}^p$).

The next lemma quantifies the ``cost`` of using the perturbed point $(a^*_{n, \lambda_n, m}, b_{n, 0} + m)$ as a candidate solution for the Version I of the optimization problem instead of the point $(a^*_{n, \lambda_n, 0}, b_{n, 0})$.

\begin{lemma}\label{lemma:difference_of_I_and_II}
    For any $m \in \mathbb{R}^p$ such that $b_{n, 0}^\top G_n (b_{n, 0} + m) \geq 0$, we have
    \begin{align*}
        & f_{n, \lambda_n}(a^*_{n, \lambda_n, m}, b_{n, 0} + m) - f_{n, \lambda_n}(a^*_{n, \lambda_n, 0}, b_{n, 0}) \geq \min \{ C_{1n}, C_{2n} \} \| A_n^{-1} \|_2^{-2} \| G_n^{-1/2} \|_2^{-2} \| m \|^2  - \lambda_n \sqrt{p} \| m \|,
    \end{align*}
    where $C_{1n} := 0.4 \rho_n \{ \rho_2(A_n)^{-1} -  1 \}$, $C_{2n} := (\sqrt{0.6} - \sqrt{0.4})^2 \rho_n$ and $A_n :=  G_n^{-1/2} M_n G_n^{-1/2}/\rho_n$.
\end{lemma}

\begin{proof}[Proof of Lemma \ref{lemma:difference_of_I_and_II}]


Application of the Cauchy-Schwarz inequality reveals that
\begin{align*}
    a^*_{n, \lambda_n, m} = \frac{\rho_n u_n + M_{n, 0} G_n^{1/2} m}{\| \rho_n u_n + M_{n, 0} G_n^{1/2} m \|}.
\end{align*}
Consequently, for $m \in \mathbb{R}^p$,
\begin{align}\label{eq:f_difference}
\begin{split}
    & f_{n, \lambda_n}(a^*_{n, \lambda_n, m}, b_{n, 0} + m) - f_{n, \lambda_n}(a^*_{n, \lambda_n, 0}, b_{n, 0}) \\
    =& -2 \| \rho_n u_n + M_{n, 0} G_n^{1/2} m \| + 2 \rho_n + 2 b_{n, 0}^\top M_n m + m^\top M_n m + \lambda_n \{ \| b_{n, 0} + m \|_1 - \| b_{n, 0} \|_1 \},
\end{split}
\end{align}
where $M_{n, 0} := G_n^{-1/2} M_n G_n^{-1/2}$ and $\rho_n$ denotes the largest eigenvalue of this matrix. Let now $A_n := M_{n, 0}/\rho_n$ and express $m$ as $m =  G_n^{-1/2} A_n^{-1} v $ for some $v \in \mathbb{R}^p$. Our earlier assumption that $b_{n, 0}^\top G_n (b_{n, 0} + m) \geq 0$ then takes the form
\begin{align*}
    0 \leq& b_{n, 0}^\top G_n (b_{n, 0} + m) = u_n^\top (u_n + G_n^{1/2} m )
    = 1 + u_n^\top A_n^{-1} v = 1 + u_n^\top v,
\end{align*}
as $u_n$ is an eigenvector of $A_n$ corresponding to the eigenvalue $1$.

Then, we can write the first part of the right-hand side of \eqref{eq:f_difference} as
\begin{align*}
    -2 \| \rho_n u_n + M_{n, 0} G_n^{1/2} m \| + 2 \rho_n + 2 b_{n, 0}^\top M_n m + m^\top M_n m =& \rho_n \{ -2 \| u_n + A_{n} G_n^{1/2} m \|  + 2 + 2 u_n^\top A_n G_n^{1/2} m + m^\top G_n^{1/2} A_n G_n^{1/2} m \} \\
    =& \rho_n \{ -2 \| u_n + v \|  + 2 + 2 u_n^\top v + v^\top A_n^{-1} v \}.
\end{align*}

Next, we have the identity $(\| u_n + v \| - 1)^2 = 1 + 2 u_n^\top v + v^\top  - 2 \| u_n + v \| + 1$. Hence, we can write
\begin{align}\label{eq:first_part_sum_of_squares}
\begin{split}
    -2 \| \rho_n u_n + M_{n, 0} G_n^{1/2} m \| + 2 \rho_n + 2 b_{n, 0}^\top M_n m + m^\top M_n m = \rho_n \{ (\| u_n + v \| - 1)^2 +  v^\top (A_n^{-1} - I_p) v \}.
\end{split}
\end{align}

Now, as $A_n$ has its eigenvalues in $[0, 1]$, then $A_n^{-1} - I_p$ has eigenvalues in $[0, \infty)$ with the eigenvector $u_n$ corresponding to the eigenvalue $0$. Decompose then $v$ such that $v = P_n v + Q_n v$ where $P_n := u_n u_n^\top$ is the projection onto $u_n$ and $Q_n := I_p - P_n$. As such, $P_n (A_n^{-1} - I_p) = 0$ and we have
\begin{align*}
    v^\top (A_n^{-1} - I_p) v =& v^\top Q_n (A_n^{-1} - I_p) Q_n v = \| Q_n v \|^2 \frac{(Q_n v)^\top}{\| Q_n v \|} (A_n^{-1} - I_p) \frac{Q_n v}{\| Q_n v \|} \geq \| Q_n v \|^2 \rho_{p - 1}(A_n^{-1} - I_p) = \| Q_n v \|^2 \{ \rho_2(A_n)^{-1} -  1 \},
\end{align*}
where $\rho_k( B )$ denotes the $k$th largest eigenvalue of the matrix $B$.

We next derive a lower bound for the right-hand side of \eqref{eq:first_part_sum_of_squares} (call it $T_n$ in the following), dividing the treatment into two cases. Fixing an arbitrary $\varepsilon \in (0, 1/2)$, if $\| Q_n v \|^2/\| v \|^2 \geq 0.5 - \varepsilon$, then the preceding paragraph shows that
\begin{align*}
    T_n \geq& \rho_n \| Q_n v \|^2 \{ \rho_2(A_n)^{-1} -  1 \} \geq \rho_n \| v \|^2 (0.5 - \varepsilon) \{ \rho_2(A_n)^{-1} -  1 \}.
\end{align*}
Take next the complementary case $\| P_n v \|^2/\| v \|^2 > 0.5 + \varepsilon$. We consider two sub-cases (I) and (II). In the first one, we assume that $\| u_n + v \| \leq 1$. In this case we have
\begin{align}\label{eq:case_1_ineq}
    | 1 - \| u_n + v \| | > 1 - \left\{ 1 + \| v \|^2 - \sqrt{2(1 + 2 \varepsilon)} \| v \| \right\}^{1/2}.
\end{align}
To see that this holds, we first observe that the quantity inside the square root can be written as $ (1 - \| v \|)^2 + \| v \| \{ 2 - \sqrt{2(1 + 2 \varepsilon)} \} $, where $2(1 + 2 \varepsilon) < 4$, making \eqref{eq:case_1_ineq} well-defined. Now, simplification shows that \eqref{eq:case_1_ineq} is equivalent to $1 + 2 u_n^\top v + \| v \|^2 < 1 + \| v \|^2 - \sqrt{2(1 + 2 \varepsilon)} \| v \|$, which in turn is equivalent to $u_n^\top v < -\| v \| (0.5 + \varepsilon)^{1/2}$, which we next show to hold. Because $\| P_n v \|^2 = (u_n^\top v)^2$, our assumption writes as $(u_n^\top v)^2 > \| v \|^2 (0.5 + \varepsilon)$. As $\| u_n + v \| \leq 1$ implies that $u_n^\top v \leq 0$, the case $u_n^\top v > \| v \| (0.5 + \varepsilon)^{1/2}$ cannot occur in the assumption and, consequently, it must hold that $u_n^\top v < -\| v \| (0.5 + \varepsilon)^{1/2}$, finally showing that \eqref{eq:case_1_ineq} is true.

We next search still for a simpler lower bound for \eqref{eq:case_1_ineq}. Observe first that our three assumptions $(u_n^\top v)^2 > \| v \|^2 (0.5 + \varepsilon)$, $\| u_n + v \| \leq 1$ and $u_n^\top v \geq -1$ imply that $-1 \leq u_n^\top v \leq 0$ and, consequently, that $\| v \| < (0.5 + \varepsilon)^{-1/2}$. We then claim that the following inequality holds, $1 - \left\{ 1 + \| v \|^2 - \sqrt{2(1 + 2 \varepsilon)} \| v \| \right\}^{1/2} \geq \| v \| (\sqrt{0.5 + \varepsilon} - \sqrt{0.5 - \varepsilon})$. This inequality is equivalent to
\begin{align}
    1 + \| v \|^2 - \sqrt{2(1 + 2 \varepsilon)} \| v \| \leq \{ 1 - \| v \| (\sqrt{0.5 + \varepsilon} - \sqrt{0.5 - \varepsilon}) \}^2,
\end{align}
which is a quadratic in $\| v \|$ and can easily be verified to hold when $\| v \| < (0.5 + \varepsilon)^{-1/2}$. Consequently in the sub-case (I), we have for $T_n$ the bound $T_n \geq \rho_n \| v \|^2 (\sqrt{0.5 + \varepsilon} - \sqrt{0.5 - \varepsilon})^2$.

What remains, is tackling the case where $\| P_n v \|^2/\| v \|^2 > 0.5 + \varepsilon$ and $\| u_n + v \| > 1$. In this case it holds that $| \| u_n + v \| - 1 | \geq \sqrt{0.5 + \varepsilon} \| v \|$. To see this, we note that the inequality is equivalent to having $2 y \geq ( C - 1 ) x^2 + 2 \sqrt{C} x$, when $C = 0.5 + \varepsilon$ and $x, y \in \mathbb{R}$ satisfy $y^2 > C x^2$, $x \geq 0$, $2y > -x^2$ and $y \geq -1$. 
Consequently, in the sub-case (II), we have the lower bound $T_n \geq \rho_n (0.5 + \varepsilon) \| v \|^2$.

Combining now all three lower bounds with the choice $\varepsilon = 0.1$ and observing that $(\sqrt{0.6} - \sqrt{0.4})^2 < 0.6$ and that $\| A_n^{-1} \|_2^{-1} \| G_n^{-1/2} \|_2^{-1} \| m \| \leq \| A_n G_n^{1/2} m \|$,  we get the desired claim.




\end{proof}

Assume now that $\lambda_n \rightarrow 0$. Our next result shows that, when $n$ is large enough, the minimizer $b_{n, \lambda_n}$ is increasingly restricted to a small neighbourhood of $b_{n, 0}$. In the result we let $\mathcal{A}_n$ denote the event that
\begin{align*}
    \mathcal{A}_n := \{ \alpha_n > 0.5 \alpha, \Omega_n > 0.5 \Omega, 1.5 \rho > \rho_n > 0.5 \rho, \rho_n - \psi_n > 0.5(\rho - \psi), \psi_n < 1.5 \psi \}.
\end{align*}
The convergences $M_n \rightarrow_p M, G_n \rightarrow_p G$ then guarantee that $P(\mathcal{A}_n) \rightarrow 1$ as $n \rightarrow \infty$.

\begin{lemma}\label{lem:minimizer_of_I}
    Assume that $\lambda_n \rightarrow 0$. Fix $\varepsilon > 0$ and choose $n_0$ such that $\lambda_n \leq \varepsilon$ for all $n \geq n_0$. Then, for $n \geq n_0$, we have the implication,
    \begin{align*}
        \mathcal{A}_n \Rightarrow \left\{ \| b_{n, \lambda_n} - b_{n, 0} \| \leq \frac{2 \varepsilon \sqrt{p}}{C} \right\},
    \end{align*}
    where $C$ is a strictly positive constant depending only on $\rho, \psi, \Omega, \alpha$.
\end{lemma}

\begin{proof}[Proof of Lemma \ref{lem:minimizer_of_I}]
    In the notation of Lemma \ref{lemma:difference_of_I_and_II}, under the event $\mathcal{A}_n$, we have that
    \begin{align*}
        & \min \{ C_{1n}, C_{2n} \} \| A_n^{-1} \|_2^{-2} \| G_n^{-1/2} \|_2^{-2} \geq \min\left\{  \frac{0.2 \rho(\rho - \psi)}{3 \psi}, \frac{1}{2} \rho \right\} \frac{1}{18 \rho^2} \Omega^2 \alpha =: C,
    \end{align*}
    where $C > 0$. Then, for $n \geq n_0$, if $\mathcal{A}_n$ holds, we have, by Lemma \ref{lemma:difference_of_I_and_II}, that $f_{n, \lambda_n}(a^*_{n, \lambda_n, m}, b_{n, 0} + m) - f_{n, \lambda_n}(a^*_{n, \lambda_n, 0}, b_{n, 0}) \geq \| m \| ( C \| m \|  - \varepsilon \sqrt{p} )$. Now this implies that any candidate minimizer $b_{n, 0} + m$ of Version I must have $\| m \| < (2 \varepsilon \sqrt{p})/C $. This is because, if $\| m \| \geq (2 \varepsilon \sqrt{p})/C $, we have $f_{n, \lambda_n}(a^*_{n, \lambda_n, m}, b_{n, 0} + m) - f_{n, \lambda_n}(a^*_{n, \lambda_n, 0}, b_{n, 0}) \geq 2 \varepsilon^2 p/C > 0$, showing that any such $b_{n, 0} + m$ is not a minimizer of the objective function of Version I. Thus the result is proven.
\end{proof}

The consistency $b_{n, \lambda_n} - b_{n, 0} \rightarrow_p 0$ now follows from Lemma \ref{lem:minimizer_of_I} and the earlier fact that $P(\mathcal{A}_n) \rightarrow 1$ as $n \rightarrow \infty$.

\begin{corollary}\label{cor:from_I_to_II}
    Assume that $\lambda_n \rightarrow 0$. Then we have $\| b_{n, \lambda_n} - b_{n, 0} \| = o_p(1)$, as $n \rightarrow \infty$.
\end{corollary}

Let now $a_n, \lambda_n \rightarrow 0$ be such that $\lambda_n/a_n \rightarrow 0$. Applying Lemma \ref{lem:minimizer_of_I} to such a sequence $\lambda_n$ and with the choice $\varepsilon = \delta a_n$, for $\delta > 0$, yields the following stronger result

\begin{corollary}\label{cor:from_I_to_II_strong}
    Assume that $a_n, \lambda_n \rightarrow 0$ be such that $\lambda_n/a_n \rightarrow 0$. Then we have $a_n^{-1} \| b_{n, \lambda_n} - b_{n, 0} \| = o_p(1)$, as $n \rightarrow \infty$.
\end{corollary}

For the next result, we make the assumption that $G_n, M_n$ are scatter matrices. That is, they are a functions $G_n \equiv G_n(x_1, \ldots, x_n)$, $M_n \equiv M_n(x_1, \ldots, x_n)$ of a data set $x_1, \ldots, x_n$ and obey the transformation rules (affine equivariance)
\begin{align*}
    G_n(A x_1, \ldots, A x_n) = A G_n(x_1, \ldots, x_n) A^\top \quad \mbox{and} \quad M_n(A x_1, \ldots, A x_n) = A M_n(x_1, \ldots, x_n) A^\top,
\end{align*}
for any invertible $A \in \mathbb{R}^{p \times p}$.

\begin{lemma}\label{lemma:from_II_to_III}
    Assume that $c_n (G_n - G) = \mathcal{O}_p(1)$ and that $c_n (M_n - M) = \mathcal{O}_p(1)$ for some increasing sequence $c_n$. Then,
    \begin{align*}
        c_n \| a_{n, 0} - \alpha \| = \mathcal{O}_p\left( 1 \right) \quad \mbox{and} \quad
        c_n \| b_{n, 0} - \beta \| = \mathcal{O}_p\left( 1 \right).
    \end{align*}
\end{lemma}

\begin{proof}[Proof of Lemma \ref{lemma:from_II_to_III}]
    By the affine equivariance of $G, M$, the solution vector $b_{n, 0}$ of Version II transforms as $b_{n, 0} \mapsto (A^{-1})^\top b_{n, 0} $ under the transformation $(x_1, \ldots, x_n) \mapsto (A x_1, \ldots, A x_n)$ \citep{tyler2009invariant}. The equivalent also holds for the solution vector of Version III. Now, as $\| b_{n, 0} - \beta \| \leq \| A^\top \|_2 \| (A^{-1})^\top b_{n, 0} - (A^{-1})^\top \beta \|$,
    where we recall that $\| \cdot \|_2$ denotes the spectral norm, the result $\| b_{n, 0} - \beta \| = \mathcal{O}_p( 1/c_n )$ follows once we show that $\| (A^{-1})^\top b_{n, 0} - (A^{-1})^\top \beta \| = \mathcal{O}_p( 1/c_n )$ for any single choice of $A$.
    
    Choose then $A = O^\top G^{-1/2}$ where $O$ is any orthogonal matrix containing the eigenvectors of $G^{-1/2} M G^{-1/2}$ as its columns. Calling the diagonal matrix of respective eigenvalues $\Lambda$, we thus have $A G A^\top = I_p$ and $A M A^\top = \Lambda$. Consequently, we may in the following assume that $G = I_p$ and that $M$ is diagonal with strictly positive diagonal elements. 
    
    Let next $u_n$ be the leading eigenvector of the matrix $G_n^{-1/2} M_n G_n^{-1/2}$ and let $u = e_1$ be the leading eigenvector of the matrix $G^{-1/2} M G^{-1/2}$. Then, Lemma 2 and Corollary 1 in \cite{virta2021determining} imply that $u_{n1}^2 - 1 = \mathcal{O}_p(1/c_n^2)$ and that $u_{nj} = \mathcal{O}_p(1/c_n)$, since we have assumed that the leading eigenvalue of $G^{-1/2} M G^{-1/2}$ is unique. By writing,
    \begin{align*}
        \sqrt{1 + \mathcal{O}_p(1/c_n^2)} - 1 = \frac{\mathcal{O}_p(1/c_n^2)}{\sqrt{1 + \mathcal{O}_p(1/c_n^2)} + 1} = \mathcal{O}_p(1/c_n^2),
    \end{align*}
    we obtain, correcting the sign of $u_n$ if necessary, that $u_{n1} - 1 = \mathcal{O}_p(1/c_n^2)$ and, consequently, that $ \| u_n - e_1 \| = \mathcal{O}_p (1/c_n) $.
     By Corollary \ref{cor:gee_via_optimization}, we have $b_{n, 0} = G_n^{-1/2} u_n$ and $\beta = e_1$. Arguing as in the proof of Lemma 1 in \cite{virta2021determining}, we obtain $\| G_n^{-1/2} - I_p \|_2 = \mathcal{O}_p(1/c_n)$, finally giving us,
     \begin{align*}
         \| b_{n, 0} - \beta \| \leq& \| b_{n, 0} - G_n^{-1/2} e_1 \| + \|  G_n^{-1/2} e_1 - e_1 \| \leq \| G_n^{-1/2} \|_2 \| b_{n, 0} - e_1 \| + \|  G_n^{-1/2} - I_p \|_2 = \mathcal{O}_p(1/c_n).
     \end{align*}
     
    Move next back to the case of general $G$ and $M$. Now, $a_{n, 0} = G_n^{1/2} b_{n, 0}$ and the convergence of $a_{n, 0}$ then follows from the convergence of $G_n$ and $b_{n, 0}$ using the triangle inequality and similar arguments as before. 
\end{proof}

Combining Corollary \ref{cor:from_I_to_II_strong} and Lemma \ref{lemma:from_II_to_III} via the triangle inequality, we obtain the main result of this section.

\begin{theorem}\label{theo:main_result}
    Assume that $c_n (G_n - G) = \mathcal{O}_p(1)$ and that $c_n (M_n - M) = \mathcal{O}_p(1)$ for some increasing sequence $c_n$. Let $a_n, \lambda_n \rightarrow 0$ be such that $\lambda_n/a_n \rightarrow 0$ and that $a_n c_n = \mathcal{O}(1)$. Then $c_n \| b_{n, \lambda_n} - \beta \| = \mathcal{O}_p(1)$.
\end{theorem}

\section{Proof of Theorem \ref{theo:main_2}}\label{sec:proofs_2}

This section contains a collection of auxiliary results from which the proof of Theorem \ref{theo:main_2} follows. In the following $\mathcal{S}_+^p$ denotes the set of $p \times p$ positive definite matrices and $\phi_k(A)$ denotes the $k$th largest eigenvalue of the matrix $A$. Throughout this appendix, the rows of a matrix $Y \in \mathbb{R}^{p \times k}$ are denoted by $y_j \in \mathbb{R}^k$, $j = 1,\ldots,p$.

\begin{theorem}\label{theorem2_res1}

Let $f(X, Y; S, R) = \| S R - X Y^\top R \|_F^2 + \lambda\sum_{j=1}^k \| y_j \|_1$, where $S, R \in \mathcal{S}_+^p$, $X, Y \in \mathbb{R}^{p \times k}$ such that $X^\top X = I_k$ and $\lambda \geq 0$ is fixed. Fix the matrices $S_0, R_0 \in \mathcal{S}_+^p$ and let $(X_0, Y_0)$ be any minimizer of $(X, Y) \mapsto f(X, Y; S_0, R_0)$. Assume that the extreme eigenvalues of $RR^\top$ satisfy $\phi_1(RR^\top) \leq 1/\varepsilon$ and $\phi_p(RR^\top) \geq \varepsilon$ for some $\varepsilon > 0$. Fix then $C_1, C_2 > 0$, let $(S, R)$ be such that $\| S - S_0 \|_F \leq C_1$ and $\| RR^\top - R_0 R_0^\top \|_F \leq C_2$, and define $M := 4 \varepsilon^{-1} \sqrt{k} (\| R_0R_0^\top \|_F + C_2) (\| S_0 \|_F + C_1)$. Then, if
\begin{align*}
    \|Y - Y_0\|_F \geq \max \left\{ 2 \sqrt{p} \varepsilon^{-2} \|Y_0\|_F + M, 2 \varepsilon^{-1} M^{-1} \left( \frac{1}{2}\varepsilon M  + 2 \sqrt{k} (\| R_0R_0^\top \|_F + C_2) (\| S_0 \|_F + C_1) +  \lambda \sqrt{pk} \right) \| Y_0 \|_F + 1 \right\},
\end{align*}
it holds that $f(X, Y; S, R) > f(X_0, Y_0; S, R)$ for all $X$.
\end{theorem}

\begin{proof}[Proof of Theorem \ref{theorem2_res1}]
    Denoting $\sum_{j=1}^k \| y_{j} \|_1 =: F(Y)$ and using the Cauchy-Schwarz inequality for matrices, we obtain
    \begin{align*}
        f(X, Y; S, R) - f(X_0, Y_0; S, R) =& \mathrm{tr}\{RR^\top (YY^\top - Y_0 Y_0^\top)\} + 2 \mathrm{tr} \{ RR^\top S (X_0 Y_0^\top - X Y^\top) \} + \lambda \{ F(Y) - F(Y_0) \} \\
        \geq& \mathrm{tr}\{RR^\top (YY^\top - Y_0 Y_0^\top)\} - 2 \| RR^\top S \|_F \| X_0 Y_0^\top - X Y^\top \|_F - \lambda F(Y_0).
    \end{align*}
    The triangle inequality and the sub-multiplicativity of the Frobenius norm then give
    \begin{align*}
        & f(X, Y; S, R) - f(X_0, Y_0; S, R) \geq \mathrm{tr}\{RR^\top (YY^\top - Y_0 Y_0^\top)\} - 2 \sqrt{k} \| RR^\top \|_F \| S \|_F ( \| Y_0 \|_F + \| Y \|_F) - \lambda F(Y_0),
    \end{align*}
    where we have used the fact that $\| X_0 \|_F = \sqrt{k}$ and $\| X \|_F = \sqrt{k}$. Observing now that our assumptions guarantee that $\| S \|_F \leq \| S_0 \|_F + C_1$ and $\| RR^\top \|_F \leq \| R_0R_0^\top \|_F + C_2$, we get 
    \begin{align*}
        f(X, Y; S, R) - f(X_0, Y_0; S, R)
        \geq \mathrm{tr}\{RR^\top (YY^\top - Y_0 Y_0^\top)\} - 2 \sqrt{k} (\| R_0R_0^\top \|_F + C_2) (\| S_0 \|_F + C_1) ( \| Y_0 \|_F + \| Y \|_F) - \lambda F(Y_0).
    \end{align*}
Now, this can be written as
    \begin{align*}
        & \mathrm{tr}\{RR^\top ((Y - Y_0 + Y_0)(Y - Y_0 + Y_0)’ - Y_0 Y_0^\top)\} - 2 \sqrt{k} (\| R_0R_0^\top \|_F + C_2) (\| S_0 \|_F + C_1) ( \| Y_0 \|_F + \| Y \|_F) - \lambda F(Y_0) \\
        =& \mathrm{tr}\{RR^\top (Y - Y_0)(Y - Y_0)’ \} + 2 \mathrm{tr}\{RR^\top (Y - Y_0)Y_0’ \} - 2 \sqrt{k} (\| R_0R_0^\top \|_F + C_2) (\| S_0 \|_F + C_1) ( \| Y_0 \|_F + \| Y \|_F) - \lambda F(Y_0).
    \end{align*}
We next use (a) Lemma 3 in \cite{schott2006high} which says that $\mathrm{tr}(AB) \geq \phi_{p}(A) \mathrm{tr}(B)$ for positive definite $A \in \mathbb{R}^{p \times p}$ and positive semi-definite $B \in \mathbb{R}^{p \times p}$, and (b) the inequality $\mathrm{tr}(AB) \geq - \sqrt{p} \phi_1(A) \| B \|_F $, following from the Cauchy-Schwarz inequality and the equivalence of norms in finite-dimensional spaces. Then, our assumptions give us that $f(X, Y; S, R) - f(X_0, Y_0; S, R)$ is bounded below by
\begin{align*}
        & \phi_p(RR^\top) \|Y - Y_0\|_F^2 - 2 \sqrt{p} \phi_1(RR^\top)\| (Y - Y_0)Y_0’ \|_F  - 2 \sqrt{k} (\| R_0R_0^\top \|_F + C_2) (\| S_0 \|_F + C_1) ( \| Y_0 \|_F + \| Y \|_F) - \lambda F(Y_0) \\
        \geq& \varepsilon \|Y - Y_0\|_F^2 - 2 \sqrt{p} \varepsilon^{-1} \|Y - Y_0\|_F\|Y_0\|_F  - 2 \sqrt{k} (\| R_0R_0^\top \|_F + C_2) (\| S_0 \|_F + C_1) ( \| Y_0 \|_F + \| Y \|_F) - \lambda \sqrt{pk} \| Y_0 \|_F \\
        \geq& \varepsilon\|Y - Y_0\|_F (\|Y - Y_0\|_F - 2 \sqrt{p} \varepsilon^{-2} \|Y_0\|_F)  - 2 \sqrt{k} (\| R_0R_0^\top \|_F + C_2) (\| S_0 \|_F + C_1) ( \| Y_0 \|_F + \| Y \|_F) - \lambda \sqrt{pk} \| Y_0 \|_F \\
        \geq& \varepsilon\|Y - Y_0\|_F M  - 2 \sqrt{k} (\| R_0R_0^\top \|_F + C_2) (\| S_0 \|_F + C_1) ( \| Y_0 \|_F + \| Y \|_F) - \lambda \sqrt{pk} \| Y_0 \|_F \\
        \geq& \frac{1}{2}\varepsilon\|Y - Y_0\|_F M + \frac{1}{2}\varepsilon( \|Y\|_F - \|Y_0\|_F) M  - 2 \sqrt{k} (\| R_0R_0^\top \|_F + C_2) (\| S_0 \|_F + C_1) ( \| Y_0 \|_F + \| Y \|_F ) - \lambda \sqrt{pk} \| Y_0 \|_F \\
        =& \frac{1}{2}\varepsilon\|Y - Y_0\|_F M - \left( \frac{1}{2}\varepsilon M  + 2 \sqrt{k} (\| R_0R_0^\top \|_F + C_2) (\| S_0 \|_F + C_1) + \lambda \sqrt{pk} \right) \| Y_0 \|_F \geq \frac{1}{2}\varepsilon M.
    \end{align*}
\end{proof}

Theorem \ref{theorem2_res1} essentially states that when the pair $(S, R)$ is close to the pair $(S_0, R_0)$, no $Y$ that is too far away from $Y_0$ can be part of a minimizing pair of $f$. Presenting the result from the viewpoint of minimizers then instantly leads to the following corollary.

\begin{corollary}\label{theorem2_res2}
    Let $f(X, Y; S, R) = \| S R - X Y^\top R \|_F^2 + \lambda\sum_{j=1}^k \| y_j \|_1$, where $S, R \in \mathcal{S}_+^p$, $X, Y \in \mathbb{R}^{p \times k}$ such that $X^\top X = I_k$ and $\lambda \geq 0$ is fixed. Fix the matrices $S_0, R_0 \in \mathcal{S}_+^p$ and let $(X(S, R), Y(S, R))$ denote any minimizer of $(X, Y) \mapsto f(X, Y; S, R)$. Then, for all $S, R, S_0, R_0 \in \mathcal{S}_+^p$ such that $\phi_1(RR^\top) \leq 1/\varepsilon < \infty$ and $\phi_p(RR^\top) \geq \varepsilon > 0$, we have
    \begin{align*}
        & \| Y(S, R) - Y(S_0, R_0) \|_F \\
        \leq& \max\left\{ 2 \sqrt{p} \varepsilon^{-2} \|Y_0\|_F + M, \frac{2}{\varepsilon M} \left( \frac{1}{2}\varepsilon M  + 2 \sqrt{k} (\| R_0 R_0^\top \|_F + \| R R^\top - R_0 R_0^\top \|_F) (\| S_0 \|_F + \| S - S_0 \|_F) +  \lambda \sqrt{pk} \right) \| Y_0 \|_F + 1 \right\},
    \end{align*}
    where $M := 4 \varepsilon^{-1} \sqrt{k} (\| R_0 R_0^\top \|_F + \| R R^\top - R_0 R_0^\top \|_F) (\| S_0 \|_F + \| S - S_0 \|_F)$.
    \begin{align*}
\end{align*}
\end{corollary}

Corollary \ref{theorem2_res2} allows us to formulate a robustness result which states that if we take a set $\mathcal{G}$ of pairs $(S, R)$ which are sufficiently well-behaving and not arbitrarily far from a ``base'' pair $(S_0, R_0) \in \mathcal{G}$, then also the corresponding minimizers of $f$ are within a finite distance of each other, uniformly in $(S, R)$.

\begin{theorem}\label{theorem2_res3}
    Let $f(X, Y; S, R) = \| S R - X Y^\top R \|_F^2 + \lambda\sum_{j=1}^k \| y_j \|_1$, where $S, R \in \mathcal{S}_+^p$, $X, Y \in \mathbb{R}^{p \times k}$ such that $X^\top X = I_k$ and $\lambda \geq 0$ is fixed. Take a fixed set $\mathcal{G} \subset \mathcal{S}_+^p \times \mathcal{S}_+^p$ and an arbitrary element $(S_0, R_0) \in \mathcal{G}$. Assume that the following hold.
    \begin{itemize}
        \item[(i)] There exists $\varepsilon > 0$ such that $\inf_{(S, R) \in \mathcal{G}} \phi_p(RR^\top) \geq \varepsilon$ and $\sup_{(S, R) \in \mathcal{G}} \phi_1(RR^\top) \leq 1/\varepsilon$.
        \item[(ii)] We have $\sup_{(S, R) \in \mathcal{G}} \| S - S_0 \|_F < \infty$ and  $\sup_{(S, R) \in \mathcal{G}} \| R R^\top - R_0 R_0^\top \|_F < \infty$.
    \end{itemize}
    Then, letting $(X(S, R), Y(S, R))$ denote any minimizer of $(X, Y) \mapsto f(X, Y; S, R)$, we have
    \begin{align*}
        \sup_{(S, R) \in \mathcal{G}} \| Y(S, R) - Y(S_0, R_0) \|_F < \infty.
    \end{align*}
\end{theorem}

\begin{proof}[Proof of Theorem \ref{theorem2_res3}]
    Take arbitrary $(S, R) \in \mathcal{G}$ and denote $M_S := \sup_{(S, R) \in \mathcal{G}} \| S - S_0 \|_F$ and $M_R := \sup_{(S, R) \in \mathcal{G}} \| R R^\top - R_0 R_0^\top \|_F$. Then, Corollary \ref{theorem2_res2} states that $\| Y(S, R) - Y(S_0, R_0) \|_F$ is bounded above by a maximum of two quantities, called hereafter $W_1$ and $W_2$. Inspecting them separately shows that
    \begin{align*}
        W_1 =& 2 \sqrt{p} \varepsilon^{-2} \|Y_0\|_F + 4 \varepsilon^{-1} \sqrt{k} (\| R_0 R_0^\top \|_F + \| R R^\top - R_0 R_0^\top \|_F) (\| S_0 \|_F + \| S - S_0 \|_F) \\
        \leq& 2 \sqrt{p} \varepsilon^{-2} \|Y_0\|_F + 4 \varepsilon^{-1} \sqrt{k} (\| R_0 R_0^\top \|_F + M_R) (\| S_0 \|_F + M_S),
    \end{align*}
    and
    \begin{align*}
        W_2 =& 2 \| Y_0 \|_F + 2 \varepsilon^{-1} M^{-1} \lambda \sqrt{pk} \| Y_0 \|_F + 1 =  2 \| Y_0 \|_F + \frac{2 \varepsilon^{-1} \lambda \sqrt{pk} \| Y_0 \|_F}{4 \varepsilon^{-1} \sqrt{k} (\| R_0 R_0^\top \|_F + \| R R^\top - R_0 R_0^\top \|_F) (\| S_0 \|_F + \| S - S_0 \|_F)} + 1 \\
        \leq& 2 \| Y_0 \|_F + \frac{\lambda \sqrt{pk} \| Y_0 \|_F}{ 2 \sqrt{k} \| R_0 R_0^\top \|_F \| S_0 \|_F} + 1.
    \end{align*}
    The above are valid for all $(S, R) \in \mathcal{G}$, showing that there exists $W_3$ such that $\| Y(S, R) - Y(S_0, R_0) \|_F \leq W_3$ for all $(S, R) \in \mathcal{G}$, from which the claim follows.
\end{proof}

Theorem \ref{theo:main_2} in the main text is now a direct restatement of Theorem \ref{theorem2_res3} in terms of the SICS problem. Note also that $ \phi_1(RR^\top) = \phi_1(RR^\top - R_0 R_0^\top + R_0 R_0^\top) \leq  \| RR^\top - R_0 R_0^\top \|_2 + \phi_1( R_0 R_0^\top) \leq \| RR^\top - R_0 R_0^\top \|_F + \phi_1( R_0 R_0^\top) $, showing that $\sup_{(S, R) \in \mathcal{G}} \| R R^\top - R_0 R_0^\top \|_F < \infty$ implies the condition $\sup_{(S, R) \in \mathcal{G}} \phi_1(RR^\top) < \infty$, letting us omit it in the main text.

The next theorem gives us a way to state an assumption in Theorem \ref{theo:main_2} with the actual scatter matrix instead of its inverse square root.
\begin{theorem}\label{theorem2_res4}
    Let $S_0, S \in \mathcal{S}_+^p$ be fixed matrices. Let $\varepsilon > 0$ be such that $\phi_p(S) \geq \varepsilon$. Then,
    \begin{align*}
        \| S^{-\frac{1}{2}} - S_0^{-\frac{1}{2}} \|_F \leq p^2 \left(\phi_p(S_0)\left(\sqrt{\varepsilon}+\varepsilon\frac{1}{\sqrt{\phi_1(S_0)}}\right)\right)^{-1}\| S-S_0 \|_F.
    \end{align*}
\end{theorem}
\begin{proof}[Proof of Theorem \ref{theorem2_res4}]
    By a direct calculation $(S^{-1}-S_0^{-1})S = -S_0^{-1}(S-S_0)$ and $S^{-1}-S_0^{-1} = (S^{-\frac{1}{2}}-S_0^{-\frac{1}{2}})S^{-\frac{1}{2}} + S_0^{-\frac{1}{2}}(S^{-\frac{1}{2}}-S_0^{-\frac{1}{2}})$. Combining those, we get
    \begin{align*}
        (S^{-\frac{1}{2}}-S_0^{-\frac{1}{2}})S^{\frac{1}{2}} + S_0^{-\frac{1}{2}}(S^{-\frac{1}{2}}-S_0^{-\frac{1}{2}})S = -S_0^{-1}(S-S_0).
    \end{align*}
    By vectorizing the matrices on both sides and using the rule $\mathrm{vec}(AXB^\top)=(B \otimes A)\mathrm{vec}(X)$, we get
    \begin{align*}
        (S^{\frac{1}{2}} \otimes I + S \otimes S_0^{-\frac{1}{2}})\mathrm{vec}(S^{-\frac{1}{2}}-S_0^{-\frac{1}{2}}) = -(I \otimes S_0^{-1})\mathrm{vec}(S-S_0).
    \end{align*}
    Now, based on the assumptions, $S^{\frac{1}{2}} \otimes I + S \otimes S_0^{\frac{1}{2}}$ is invertible as all the individual matrices have positive eigenvalues. Therefore we have
    \begin{align*}
        \mathrm{vec}(S^{-\frac{1}{2}}-S_0^{-\frac{1}{2}}) = -(S^{\frac{1}{2}} \otimes I + S \otimes S_0^{-\frac{1}{2}})^{-1}(I \otimes S_0^{-1})\mathrm{vec}(S-S_0),
    \end{align*}
    which implies
    \begin{align*}
        \|\mathrm{vec}(S^{-\frac{1}{2}}-S_0^{-\frac{1}{2}})\| &= \| (S^{\frac{1}{2}} \otimes I + S \otimes S_0^{-\frac{1}{2}})^{-1}(I \otimes S_0^{-1})\mathrm{vec}(S-S_0) \| \\
        \|S^{-\frac{1}{2}}-S_0^{-\frac{1}{2}}\|_F &\leq \| (S^{\frac{1}{2}} \otimes I + S \otimes S_0^{-\frac{1}{2}})^{-1}(I \otimes S_0^{-1}) \|_F \|S-S_0 \|_F \\
        &\leq p^2 \left(\phi_p(S_0)\left(\sqrt{\varepsilon}+\varepsilon\frac{1}{\sqrt{\phi_1(S_0)}}\right)\right)^{-1}\| S-S_0 \|_F,
    \end{align*}
    where we have used the sub-multiplicativity of the Frobenius norm, the fact that the eigenvalues of $C_1 \otimes C_2$ for $C_1, C_2 \in \mathcal{S}_+^p$ are all pairwise products of the eigenvalues of $C_1$ and $C_2$, and Weyl's inequality, $\phi_p(C_1 + C_2) \geq \phi_p(C_1) + \phi_p(C_2)$.
\end{proof}

\bibliographystyle{myjmva}
\bibliography{references}
\end{document}